\documentclass[11pt,letterpaper]{article}

\usepackage[letterpaper, margin=1in]{geometry}
\usepackage{authblk}
\usepackage{enumitem}

\usepackage{amsmath, commath, amsfonts, amssymb, amsthm}
\usepackage{commath}
\usepackage[ruled,vlined]{algorithm2e}

\usepackage[linewidth=1pt]{mdframed}
\usepackage{tikz}
\usetikzlibrary{arrows.meta,shapes,calc}
\usepackage[labelformat=simple]{subcaption}

\usepackage{xparse}
\usepackage{ifthen}

\usepackage[colorlinks,allcolors=blue]{hyperref}
\usepackage[square,sort,numbers]{natbib}
\title{
  Byzantine Consensus in Directed Hypergraphs\footnote{
    This research is supported in part by the National Science Foundation
    award 1733872.
    Any opinions, findings, and conclusions or recommendations
    expressed here are those of the authors and do not necessarily
    reflect the views of the funding agencies or the U.S. government.
  }
}

\author[1]{Muhammad Samir Khan}
\affil[1]{
  Department of Computer Science\protect\linebreak
  University of Illinois at Urbana-Champaign\protect\linebreak
  \texttt{mskhan6@illinois.edu}\protect\linebreak
}

\author[2]{Nitin H. Vaidya}
\affil[2]{
  Department of Computer Science\protect\linebreak
  Georgetown University\protect\linebreak
  \texttt{nitin.vaidya@georgetown.edu}\protect\linebreak
}

\newtheorem{theorem}{Theorem}[section]
\newtheorem{lemma}[theorem]{Lemma}

\newtheorem{definition}[theorem]{Definition}

\newtheorem{claim}[theorem]{Claim}
\newtheorem{observation}[theorem]{Observation}

\newcommand{\newRef}[2]{\hyperref[#1]{#2 \ref*{#1}}}
\newcommand{\appendixRef}[1]{\newRef{#1}{Appendix}}
\newcommand{\sectionRef}[1]{\newRef{#1}{Section}}
\newcommand{\figureRef}[1]{\newRef{#1}{Figure}}

\newcommand{\algoRef}[1]{\newRef{#1}{Algorithm}}
\newcommand{\theoremRef}[1]{\newRef{#1}{Theorem}}
\newcommand{\lemmaRef}[1]{\newRef{#1}{Lemma}}
\newcommand{\obsRef}[1]{\newRef{#1}{Observation}}
\newcommand{\defRef}[1]{\newRef{#1}{Definition}}

\newcommand{\claimRef}[1]{\newRef{#1}{Claim}}

\newcommand{\pageRef}[1]{\hyperref[#1]{page \pageref*{#1}}}

\newcounter{condition}

\def\nameOfConditionIoNC{LCR-io}
\def\nameOfConditionIoSC{AB-io}
\def\nameOfConditionHNC{LCR-hyper}
\def\nameOfConditionHSC{AB-hyper}
\def\nameOfConditionPNC{LCR-p2p}
\def\nameOfConditionPSC{AB-p2p}
\def\nameOfConditionBNC{LCR-local}
\def\nameOfConditionBSC{AB-local}
\def\nameOfConditionMNC{LCR}
\def\nameOfConditionMSC{AB}

\newcommand{\conditionNC}
{{\hyperref[definition hypergraphs NC]{condition \nameOfConditionHNC}}}
\newcommand{\conditionSC}
{{\hyperref[definition hypergraphs SC]{condition \nameOfConditionHSC}}}
\newcommand{\conditionpNC}
{{\hyperref[definition p2p NC]{condition \nameOfConditionPNC}}}

\newcommand{\conditionbNC}
{{\hyperref[definition local NC]{condition \nameOfConditionBNC}}}

\newcommand{\ConditionNC}
{{\hyperref[definition hypergraphs NC]{Condition \nameOfConditionHNC}}}

\newcounter{index}

\newcommand{\propagate}[1]{\rightsquigarrow_{#1}}
\newcommand{\notpropagate}[1]{\not \rightsquigarrow_{#1}}

\newcommand{\adjacent}[1]{\rightarrow_{#1}}
\newcommand{\notadjacent}[1]{\not \rightarrow_{#1}}

\NewDocumentCommand\inneighborhood{ggg}{\ensuremath{\IfNoValueTF{#3}{\Gamma_{#1}}{\Gamma_{#3}(#1, #2)}}}
\newcommand{\floor}[1]{\left\lfloor {#1} \right\rfloor}

\newcommand{\graphSplitFSet}{\Lambda}

\newcommand{\Step}[1]{\texttt{Step (#1)}}
\newcommand{\step}[1]{\texttt{step (#1)}}

\newcommand{\underlying}[1]{\overline{#1}}

\newcommand{\head}{{H}}
\newcommand{\tail}{{T}}
\newcommand{\incidentedges}{{\delta}}

\newcommand{\suchthat}{:}

\def\arrow{-{Latex[width=2mm,length=2mm]}}
\def\otherarrow{-{Latex[width=2mm,length=2mm,open]}}

\begin{document}

\maketitle

\begin{abstract}
  Byzantine consensus is a classical problem in distributed computing.
Each node in a synchronous system starts with a binary input.
The goal is to reach agreement in the presence of Byzantine faulty nodes.
We consider the setting where communication between
nodes is modelled via a directed hypergraph.
In the classical \emph{point-to-point} communication model,
the communication between nodes is modelled as a simple graph
where
all messages sent on an edge are private between the two endpoints of the edge.
This allows a faulty node to \emph{equivocate}, i.e.,
lie differently to its different neighbors.
Different models have been proposed in the literature
that weaken equivocation.
In the \emph{local broadcast} model,
every message transmitted by a node
is received identically and correctly
by all of its neighbors.
In the \emph{hypergraph} model,
every message transmitted by a node on a hyperedge
is received identically and correctly
by all nodes on the hyperedge.
Tight network conditions are known for each of the three cases
for undirected (hyper)graphs.
For the directed models,
tight conditions are known for the
point-to-point and local broadcast models.

In this paper,
we consider the directed hypergraph model
that encompasses all the models above.
Each directed hyperedge consists of a single head (sender)
and at least one tail (receiver),
This models a \emph{local multicast} channel
where messages sent by the head node (sender)
are received identically by all the tail nodes (receivers)
in the hyperedge.
For this model,
we identify tight network conditions for consensus.
We observe how the directed hypergraph model
(which we will also refer to as the local multicast model)
reduces to each of the three models above
under specific conditions.
In each of the three cases,
we relate our network condition to the
corresponding known tight conditions.
The local multicast model
also encompasses other practical network models
of interest that have not been explored previously,
as elaborated in the paper. \end{abstract}

\section{Introduction} \label{section introduction}
Byzantine consensus is a classical problem in distributed computing
introduced by Lamport et al.
\cite{Lamport:1982:BGP:357172.357176,Pease:1980:RAP:322186.322188}.
There are $n$ nodes in a synchronous system.
Each node starts with a binary input.
At most $f < n$ of these nodes can be Byzantine faulty,
i.e., exhibit arbitrary behavior.
The goal of a consensus protocol is for the non-faulty nodes to
reach agreement on a single output value in finite time.
To exclude trivial protocols,
we require that the output must be an input of some non-faulty node.

In this paper,
we study consensus under \emph{local multicast} channels,
which can be modelled as directed hypergraphs.
A hypergraph is a generalization of graphs
consisting of nodes and hyperedges.
Unlike an edge in a graph, a hyperedge can connect any number of nodes.
In the local multicast model,
nodes are connected via a directed hypergraph $G$.
A local multicast channel is a directed hyperedge
defined by a single sender and a non-empty set of receivers.
Each node $u$ may potentially serve
as the sender on multiple local multicast channels/hyperedges.
When node $u$ sends a message on one of its local multicast channels/hyperedges,
This model generalizes the following models that have been considered
before in the literature.

\begin{enumerate}
  \item \emph{Point-to-point communication model:}
  In the classical \emph{point-to-point} communication model,
  each edge $(u, v)$ in the communication graph
  represents a private link from node $u$ to node $v$.
  This model is well-studied
  \cite{Attiya:2004:DCF:983102,DOLEV198214,Lamport:1982:BGP:357172.357176,Lynch:1996:DA:2821576,Pease:1980:RAP:322186.322188,LewisByzantineDirected,LewisByzantineDirectedARXIV}.
  It is well-known that, for \emph{undirected} graphs
  $n \ge 3f+1$ and
  node connectivity at least $2f+1$
  are both necessary and sufficient in this model.

  \item \emph{Local broadcast model:}
  Recently,
  in \cite{khan2019undirectedPODC,khan2020directedOPODIS}, we
  studied consensus under the \emph{local broadcast model}
  \cite{Bhandari:2005:RBR:1073814.1073841,Koo:2004:BRN:1011767.1011807},
  where a message sent by any node is received identically
  by all of its neighboring nodes in the communication graph.
  For \emph{undirected} graphs,
  minimum node degree at least $2f$ and
  node connectivity at least $\floor{3f/2} + 1$
  are both necessary and sufficient for
  Byzantine consensus \cite{khan2019undirectedPODC}
  under the local broadcast model.

  \item \emph{Undirected hypergraph model:}
  Communication networks modelled as \emph{undirected} hypergraphs
  have been studied in the literature
  \cite{Fitzi:2000:PCG:335305.335363,Jaffe:2012:PEB:2332432.2332491,Ravikant10.1007/978-3-540-30186-8_32}.
  A message sent by a node $u$ on an undirected hyperedge $e \ni u$
  is received identically by all nodes in $e$.
  For this model,
  Ravikant et al. \cite{Ravikant10.1007/978-3-540-30186-8_32}
  gave tight conditions for Byzantine consensus on
  $(2,3)$-hypergraphs.\footnote{
    i.e.,
    each hyperedge consists of either $2$ or $3$ nodes.
  }
  As we discuss in \sectionRef{section other models},
  these conditions extend to general undirected hypergraphs as well.
\end{enumerate}

The classical point-to-point communication model allows
a faulty node to \emph{equivocate},
i.e., send conflicting messages to its neighbors without
this inconsistency being observed by the neighbors.
For example,
a faulty node $z$ may tell its neighbor $u$ that it has input $0$,
but tell another neighbor $v$ that it has input $1$.
Since messages on each edge are private between the two endpoints,
node $u$ does not overhear the message sent to node $v$ and vice versa.
The local broadcast model and the hypgergraph model
restrict a faulty node's ability to equivocate
by detecting such attempts.
In the local broadcast model,
a faulty node's attempt to equivocate
is detected by its neighboring nodes in the communication graph.
In the undirected hypergraph model,
a faulty node's attempt to equivocate on an (undirected) hyperedge
is detected by the nodes in that hyperedge.
In our local multicast model,
a faulty node's attempt to equivocate on a single multicast channel,
i.e., on a single directed hyperedge,
is detected by the receivers in that channel.

In this work,
we introduce the \emph{local multicast} model,
that unifies the models identified above,
and make the following main contributions:

\begin{enumerate}
  \item
  \textbf{Necessary and sufficient condition for local multicast model:}
  In \sectionRef{section multiple},
  we present a network condition, and show that it is both
  necessary and sufficient for Byzantine consensus
  under the local multicast model.
  The identified condition is inspired by the
  network conditions for directed graphs
  \cite{khan2020directedOPODIS,LewisByzantineDirected},
  where node connectivity
  does not adequately capture the network requirements for consensus.
  We present a simple algorithm,
  inspired by
  \cite{khan2019undirectedPODC, khan2020directedOPODIS, LewisByzantineDirected}.

  \item
  \textbf{Reductions to the existing models:}
  The two extremes of the local multicast model
  are 1) each channel consists of exactly one receiver,
  and 2) each node has exactly one multicast channel.
  These correspond to the point-to-point communication model
  and the local broadcast model, respectively.
  In \sectionRef{section other models},
  we show how the network condition for the local multicast model
  reduces to the network requirements for
  the point-to-point model and
  the local broadcast model
  at the two extremes.
  On the other hand,
  if the hypergraph is undirected,
  then we show that the network condition reduces to the
  network requirements of the undirected hypergraph model given by
  Ravikant et al. \cite{Ravikant10.1007/978-3-540-30186-8_32}.
  Moreover, our algorithm for the local multicast model
  works for all the three models identified here as well.

  \item
  \textbf{Extensions to other models:}
  The local multicast model also captures
  some other models of practical interest
  (see \sectionRef{section new models}).
  For instance, consider the scenario where nodes are connected
  via a WiFi network.
  This can be modelled as local multicast
  over a graph $G_1$.
  Separately, the nodes are also connected via a bluetooth network,
  modelled using local multicast over a graph $G_2$
  (with the same node set as $G_1$).
  Then the union of these networks $G_1 \cup G_2$
  can be captured using the local multicast model as well.
  As another example, consider the scenario where
  nodes are connected via point-to-point channels,
  in addition to a wireless network
  with local broadcast guarantees.
  As before, this can also be captured using the local multicast model.
  Our algorithm works for these cases as well.
\end{enumerate}

In our recent work \cite{TBD_disc2021},
we obtained an analogous tight condition
for the ``bidirectional'' case when the underlying simple graph is undirected,
i.e., if a node $u$ can send messages to a node $v$,
then $v$ can also send messages to node $u$.
The tight condition obtained here is a natural extension of
the tight condition obtained in \cite{TBD_disc2021}.
However,
the results and proofs in this work
are more general and encompass the results in \cite{TBD_disc2021}. 
\section{System Model and Problem Formulation} \label{section notation}
We consider a synchronous system where nodes are
connected via a directed hypergraph $G = (V, E)$.
$V$ is the set of $n$ nodes.
Each directed hyperedge $e \in E$ is of the form $e = (u, S)$,
where $u \in V$ and $S \subseteq V - u$,
representing a \emph{local multicast channel} with sender $u$ and receivers $S$.
By convention used here, $u$ is not included in $S$.
However, trivially, each node receives its own message transmissions as well.
$u$ is the \emph{head} of $e$, denoted by $\head(e) = u$,
and each node in $S$ is a \emph{tail} of $e$,
denoted by $\tail(e) = S$.
Observe that each hyperedge has a single head and at least one tail.
For example, $(u, \set{v, w})$
is a hyperedge with head $u$ and two tail nodes $v$ and $w$.
A message $m$ sent by a node $u$ on a hyperedge $e$ (such that $\head(e) = u$)
is received identically and correctly by all tail nodes $\tail(e)$ of $e$.
Moreover, each recipient $v \in \tail(e)$ knows that $m$ was
sent by $u$ on the hyperedge $e$.
We assume that each hyperedge represents a FIFO multicast channel.

We use $\incidentedges_G(u)$ to denote the set of hyperedges in $G$
that have $u$ as the head node, i.e.,
\[\incidentedges_G(u) = \set{ e \in E(G) \mid \head(e) = u}.\]
The hypergraph $G$ has an underlying directed simple graph,
denoted by $\underlying{G} = (\underlying{V}, \underlying{E})$, such that
\begin{align*}
  \underlying{V} &:= V, \\
  \underlying{E} &:= \set{
    (u, v) \mid \exists e \in E \suchthat
      u = \head(e), v \in \tail(e)
  }.
\end{align*}

\paragraph{Neighbors:}
A node $u$ is an \emph{in-neighbor} of node $v$ in $G$
if there exists a hyperedge $e$ with $u = \head(e)$ as the head and
$v \in \tail(e)$ as one of the tails.
We call $v$ an \emph{out-neighbor} of $u$.
Note that $u$ is an in-neighbor of node $v$ in $G$
if and only if
$u$ is an in-neighbor of $v$ in $\underlying{G}$.

\begin{itemize}\item {\emph{In-neighborhood:}}
  More generally, for two disjoint sets $A, B \subseteq V(G)$,
  $\inneighborhood{A}{B}{G}$ defined below is the set of in-neighbors of
  $B$ in $A$.
  \[
    \inneighborhood{A}{B}{G} = \set{u \in A \mid \exists v \in B \suchthat \text{$u$ is an in-neighbor of $v$ in $G$}}.
  \]

  \item {\emph{Adjacent:}}
  We use $A \adjacent{G} B$
  (read as $A$ is ``adjacent'' to $B$ in $G$)
  to denote that either
  \begin{enumerate}[label=(\roman*)]
    \item $B = \emptyset$, or
    \item nodes in $B$ have at least $f+1$ in-neighbors in $A$ in $G$,
    i.e., \[\abs{ \inneighborhood{A}{B}{G} } \ge  f + 1.\]
  \end{enumerate}
\end{itemize}

\paragraph{Node split:}
We now introduce the notion of a \emph{node split}
that is used to specify the necessary and sufficient condition
under the local multicast model.
As seen later,
we will use the notion of node split to simulate possible equivocation
by a faulty node.
Intuitively,
by splitting a node $v$,
we are creating two copies of $v$ and dividing up the hyperedges
amongst the two copies.
\figureRef{figure node split} shows two examples of node split.
Formally,
splitting a set of nodes $X$ in $G$
creates a new hypergraph $G' = (V', E')$ as follows.
Each node $v \in X$ is replaced by two nodes $v^0$ and $v^1$,
so that
\[
  V' = (V - X) \cup \set{ v^0, v^1 \mid v \in X }.
\]
Consider each node $u \in V$ and
hyperedge $e = (u, S) \in \incidentedges_G(u)$.
\begin{itemize}
  \item 
  If $u \not \in X$,
  then add a hyperedge $(u, S')$ in $G'$, where
  \[S' = (S - X) \cup \set{ v^0, v^1 \mid v \in S \cap X},\]
  i.e.,
  each node $v \in \tail(e) \cap X$
  is replaced by the two nodes $v^0$ and $v^1$.

  \item
  If $u \in X$,
  then choose a node $u_e \in \set{u^0, u^1}$, and
  add a hyperedge $(u_e, S')$ in $G'$, where
  \[S' = (S - X) \cup \set{ v^0, v^1 \mid v \in S \cap X}.\]
  Recall that
$u \notin S$ by convention.
  Note that the choice of $u_e \in \set{u^0, u^1}$
  affects the set of hyperedges of hypergraph $G'$.
  For simplicity, we say that the hyperedge $e$
  has been \emph{assigned} to $u_e$.
\end{itemize}

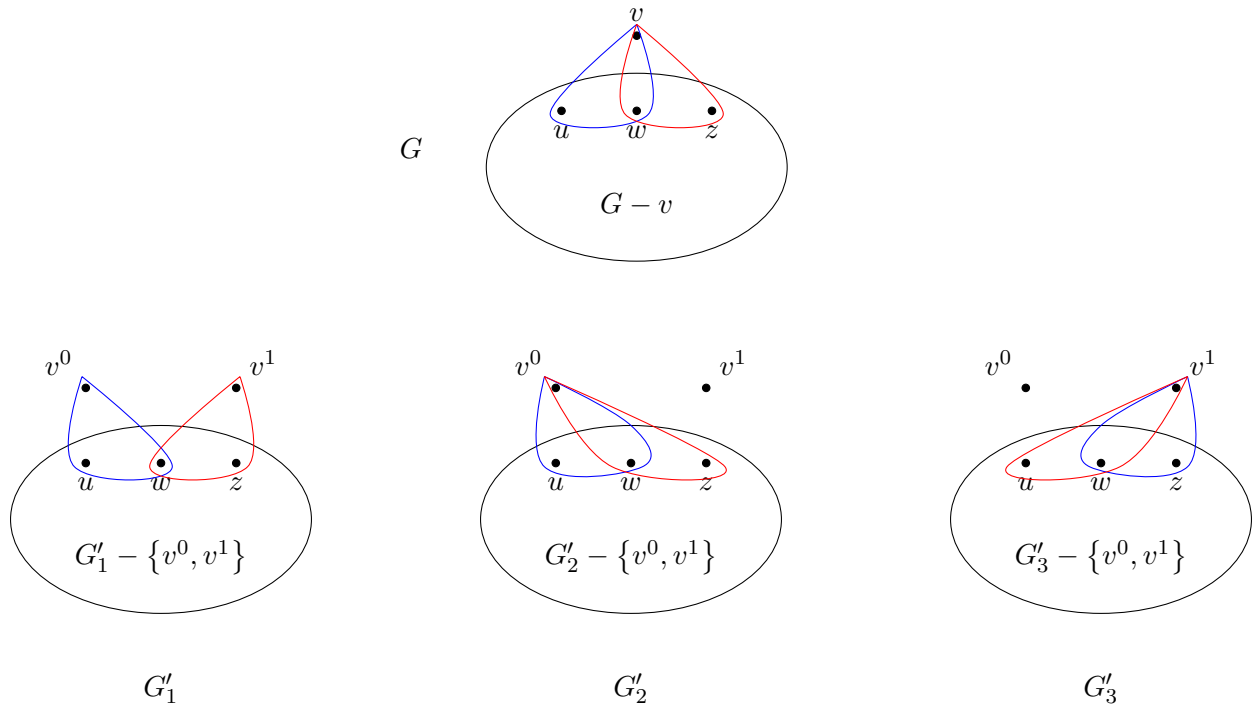
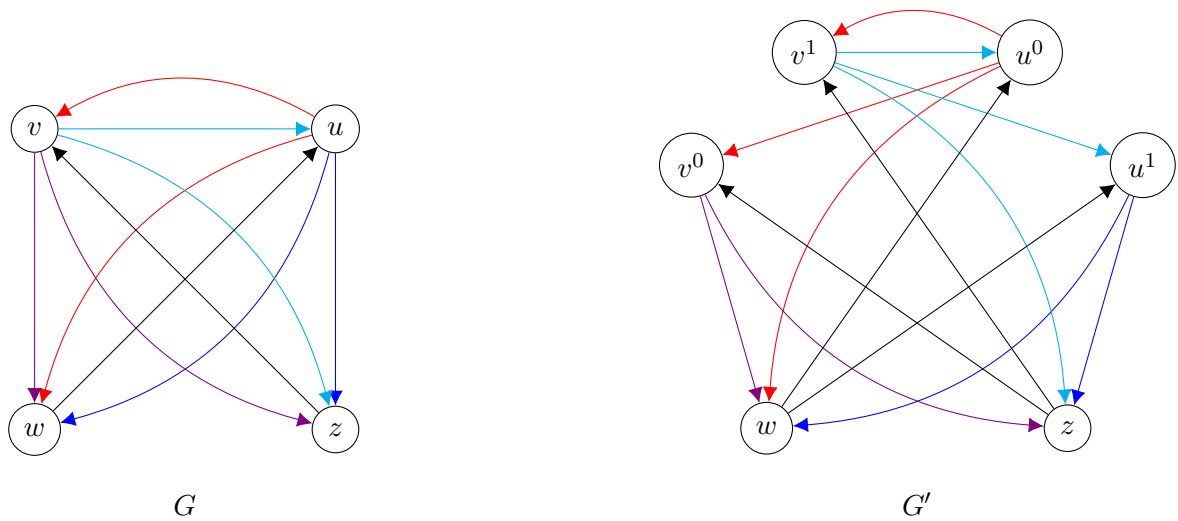
\begin{figure}[p]
  \centering
  \begin{subfigure}{\textwidth}
    \centering
    \begin{tikzpicture}
      \node[draw,circle,fill,inner sep=1pt,label={above:$v$}] at (0, 0) (v) {};
      \node[draw,circle,fill,inner sep=1pt,label={below:$u$}] at (-1, -1) (u) {};
      \node[draw,circle,fill,inner sep=1pt,label={below:$w$}] at (0, -1) (w) {};
      \node[draw,circle,fill,inner sep=1pt,label={below:$z$}] at (1, -1) (z) {};
      \node[draw,ellipse,minimum width=4cm,minimum height=2.5cm] at (0,-1.75) (G) {};
      \node at (0, -2.25) () {$G - v$};
      \node at (-3, -1.5) () {$G$};
      \node at (3, -1.5) () {};
  
      \draw[blue] plot[smooth,tension=.7] coordinates
      {(0,0.15) (-1.15,-1.05) (.15,-1.05) (0,0.15)};
      \draw[red] plot[smooth,tension=.7] coordinates
      {(0,0.15) (1.15,-1.05) (-.15,-1.05) (0,0.15)};
    \end{tikzpicture}
    \\
    \bigskip~
    \\
    \begin{tikzpicture}
      \node[draw,circle,fill,inner sep=1pt,label={above left:$v^0$}] at (-1, 0) (v0) {};
      \node[draw,circle,fill,inner sep=1pt,label={above right:$v^1$}] at (1, 0) (v1) {};
      \node[draw,circle,fill,inner sep=1pt,label={below:$u$}] at (-1, -1) (u) {};
      \node[draw,circle,fill,inner sep=1pt,label={below:$w$}] at (0, -1) (w) {};
      \node[draw,circle,fill,inner sep=1pt,label={below:$z$}] at (1, -1) (z) {};
      \node[draw,ellipse,minimum width=4cm,minimum height=2.5cm] at (0,-1.75) (G) {};
      \node at (0, -2.25) () {$G'_1-\set{v^0, v^1}$};
      \node at (0, -4) () {$G'_1$};
  
\draw[blue] plot[smooth,tension=.7] coordinates
      {(-1.05,0.15) (-1.15,-1.05) (.15,-1.05) (-1.05,0.15)};
      \draw[red] plot[smooth,tension=.7] coordinates
      {(1.05,0.15) (1.15,-1.05) (-.15,-1.05) (1.05,0.15)};
    \end{tikzpicture}
    \hfill
    \begin{tikzpicture}
      \node[draw,circle,fill,inner sep=1pt,label={above left:$v^0$}] at (-1, 0) (v0) {};
      \node[draw,circle,fill,inner sep=1pt,label={above right:$v^1$}] at (1, 0) (v1) {};
      \node[draw,circle,fill,inner sep=1pt,label={below:$u$}] at (-1, -1) (u) {};
      \node[draw,circle,fill,inner sep=1pt,label={below:$w$}] at (0, -1) (w) {};
      \node[draw,circle,fill,inner sep=1pt,label={below:$z$}] at (1, -1) (z) {};
      \node[draw,ellipse,minimum width=4cm,minimum height=2.5cm] at (0,-1.75) (G) {};
      \node at (0, -2.25) () {$G'_2-\set{v^0, v^1}$};
      \node at (0, -4) () {$G'_2$};
  
\draw[blue] plot[smooth,tension=0.7] coordinates
      {(-1.15,0.15) (-1.15,-1.05) (.15,-1.05) (0,-0.5) (-1.15,0.15)};
      \draw[red] plot[smooth,tension=.7] coordinates
      {(-1.15,0.15) (1.25,-1.05) (-0.25,-1.05) (-1.15,0.15)};
    \end{tikzpicture}
    \hfill
    \begin{tikzpicture}
      \node[draw,circle,fill,inner sep=1pt,label={above left:$v^0$}] at (-1, 0) (v0) {};
      \node[draw,circle,fill,inner sep=1pt,label={above right:$v^1$}] at (1, 0) (v1) {};
      \node[draw,circle,fill,inner sep=1pt,label={below:$u$}] at (-1, -1) (u) {};
      \node[draw,circle,fill,inner sep=1pt,label={below:$w$}] at (0, -1) (w) {};
      \node[draw,circle,fill,inner sep=1pt,label={below:$z$}] at (1, -1) (z) {};
      \node[draw,ellipse,minimum width=4cm,minimum height=2.5cm] at (0,-1.75) (G) {};
      \node at (0, -2.25) () {$G'_3-\set{v^0, v^1}$};
      \node at (0, -4) () {$G'_3$};
  
\draw[blue] plot[smooth,tension=0.7] coordinates
      {(1.15,0.15) (1.15,-1.05) (-.15,-1.05) (0,-0.5) (1.15,0.15)};
      \draw[red] plot[smooth,tension=.7] coordinates
      {(1.15,0.15) (-1.25,-1.05) (0.25,-1.05) (1.15,0.15)};
    \end{tikzpicture}
    \caption{
      Splitting a single node $v$.
      Only the hyperedges in $\incidentedges_G(v)$ are drawn here.
      There are two hyperedges in $\incidentedges_G(v)$:
      $(v, \set{u, w}))$ and $(v, \set{w, z})$,
      drawn with blue and red colors, respectively.
      The three possible hypergraphs in $\graphSplitFSet_{\set{v}}(G)$,
      other than $G$,
      corresponding to the assignment of hyperedges
      when $v$ is split into $v^0$ and $v^1$.
      These are depicted as hypergraphs $G'_1$, $G'_2$, and $G'_3$.
    }
    \label{figure node split v}
  \end{subfigure}
  \\
  \bigskip~
  \\
  \begin{subfigure}{\textwidth}
    \centering
    \begin{tikzpicture}
      \node[draw,circle] at (-2, 2) (v) {$v$};
      \node[draw,circle] at (2, 2) (u) {$u$};
      \node[draw,circle] at (-2, -2) (w) {$w$};
      \node[draw,circle] at (2, -2) (z) {$z$};
      \node at (0, -3) (G) {$G$};

      \draw[red, \arrow, bend right] (u) to (v);
      \draw[red, \arrow, bend right] (u) to (w);
      \draw[blue, \arrow, bend left] (u) to (w);
      \draw[blue, \arrow] (u) to (z);
      \draw[cyan, \arrow] (v) to (u);
      \draw[cyan, \arrow, bend left] (v) to (z);
      \draw[violet, \arrow, bend right] (v) to (z);
      \draw[violet, \arrow] (v) to (w);
      \draw[\arrow] (w) to (u);
      \draw[\arrow] (z) to (v);
    \end{tikzpicture}
    ~~~~~~~~~~~~~~~~~~~~~~~~~~~~~
    \begin{tikzpicture}
      \node[draw,circle] at (-3, 1.5) (v0) {$v^0$};
      \node[draw,circle] at (-1.5, 3) (v1) {$v^1$};
      \node[draw,circle] at (1.5, 3) (u0) {$u^0$};
      \node[draw,circle] at (3, 1.5) (u1) {$u^1$};
      \node[draw,circle] at (-2, -2) (w) {$w$};
      \node[draw,circle] at (2, -2) (z) {$z$};
      \node at (0, -3) (G) {$G'$};

      \draw[red, \arrow] (u0) to (v0);
      \draw[red, \arrow, bend right] (u0) to (v1);
      \draw[red, \arrow, bend right] (u0) to (w);
      \draw[blue, \arrow, bend left] (u1) to (w);
      \draw[blue, \arrow] (u1) to (z);
      \draw[cyan, \arrow] (v1) to (u0);
      \draw[cyan, \arrow] (v1) to (u1);
      \draw[cyan, \arrow, bend left] (v1) to (z);
      \draw[violet, \arrow, bend right] (v0) to (z);
      \draw[violet, \arrow] (v0) to (w);
      \draw[\arrow] (w) to (u0);
      \draw[\arrow] (w) to (u1);
      \draw[\arrow] (z) to (v0);
      \draw[\arrow] (z) to (v1);
    \end{tikzpicture}
    \caption{
      Splitting two nodes $u, v$ in a 4-node hypergraph $G$.
      Edges of the same color,
      which have the same head node,
      represent a single hyperedge.
      $G'$ is obtained by splitting nodes $u$ and $v$
      into $u^0, u^1$ and $v^0, v^1$, respectively.
      The cyan hyperedge is assigned to $v^1$,
      the violet hyperedge is assigned to $v^0$,
      the red hyperedge is assigned to $u^0$, and
      the blue hyperedge is assigned to $u^1$.
}
    \label{figure node split u v}
  \end{subfigure}

  \caption{
    Examples of the node split operation.
  }
  \label{figure node split}
\end{figure}

Observe that, for every node $u \in V'$ in the hypergraph $G'$,
each hyperedge in $\incidentedges_{G'}(u)$ corresponds to
a single hyperedge in $G$.
Similarly, for every node $u \in V$ in the original hypergraph $G$,
each hyperedge in $\incidentedges_G(u)$ corresponds to
a single hyperedge in $G'$.

For a set $F \subseteq V(G)$,
let $\graphSplitFSet_F(G)$ be the set of all hypergraphs
that can be obtained from $G$
by splitting some subset of nodes in the set $F$.
For a graph $G' \in \graphSplitFSet_F(G)$,
we use $F'$ to denote the set of nodes in $G'$
that correspond to nodes in $F$ in $G$,
i.e.,
\[F' := (V' \cap F) \cup (V' - V).\]

Note that there are two choices in the
node split operation above which give rise to
all the hypergraphs in $\graphSplitFSet_F(G)$:
\begin{enumerate}[label=\arabic*)]
  \item choice of which nodes in $F$ to split, and
  \item assignment of hyperedges for each split node.
\end{enumerate}
As needed, we will occasionally clarify these choices to specify how
a hypergraph $G' \in \graphSplitFSet_F(G)$ was constructed
by splitting some nodes in $F$. 
\section{Main Result} \label{section multiple}

The main result of this paper
is a tight characterization of network requirements
for Byzantine consensus under the directed hypergraph model.
Recall that for a hypergraph $G' \in \graphSplitFSet_F(G)$ obtained
by splitting some nodes in $F$,
we use $F'$ to denote the set of nodes in $G'$
that correspond to nodes in $F$ in $G$.
With a slight abuse of terminology,
we allow a partition of a set to have empty parts.

\begin{theorem} \label{theorem hypergraphs main}
  Byzantine consensus tolerating at most $f$ faulty nodes
  is achievable on a directed hypergraph $G$
  if and only if
  for every $F \subseteq V$ of size at most $f$,
  every $G' \in \graphSplitFSet_F(G)$ satisfies the following:
  for every partition $(L, C, R)$ of $V'$, either
  \begin{enumerate}[label=\arabic*)]
    \item $L \cup C \adjacent{G'} R - F'$, or
    \item $R \cup C \adjacent{G'} L - F'$.
  \end{enumerate}
\end{theorem}

Note that we allow a partition to have empty parts.
However, the interesting partitions are those
where both $L$ and $R$ are non-empty,
but $C$ can be possibly empty.
In \sectionRef{section other models},
we show that
when the directed hypergraph corresponds to
the point-to-point, local broadcast, or undirected hypergraph model,
the above condition reduces to the corresponding known tight network conditions
in each of the three cases.

We prove the necessity portion of \theoremRef{theorem hypergraphs main}
in \sectionRef{section necessity}.
In \sectionRef{section sufficiency} we give an algorithm to
constructively show the sufficiency.
The above condition is similar to the network condition
for directed graphs in the point-to-point communication model
\cite{LewisByzantineDirectedARXIV, LewisByzantineDirected}
and in the local broadcast model
\cite{khan2020directedOPODIS}.
For convenience,
we give a name to the condition in \theoremRef{theorem hypergraphs main}.

\refstepcounter{condition} \label{condition hypergraphs NC}
\begin{definition}[Condition \nameOfConditionHNC{}]
  \label{definition hypergraphs NC}
  A graph $G$ satisfies
  \emph{condition \nameOfConditionHNC{} with parameter $F$}
  if for every $G' \in \graphSplitFSet_F(G)$
  and every partition $(L, C, R)$ of $V'$,
  we have that either
  \begin{enumerate}[label=\arabic*)]
    \item $L \cup C \adjacent{G'} R - F'$, or
    \item $R \cup C \adjacent{G'} L - F'$.
  \end{enumerate}
  We say that $G$ satisfies \emph{condition \nameOfConditionHNC{}},
  if $G$ satisfies condition \nameOfConditionHNC{} with parameter $F$
  for every set $F \subseteq V(G)$ of cardinality at most $f$.
\end{definition}

\section{Reductions to Other Models} \label{section other models}

In this section, we discuss how \conditionNC{}
relates to the tight conditions for the classical point-to-point
communication model, the local broadcast model,
and the undirected hypergraph model.
In an undirected hypergraph,
any node on the undirected hyperedge can act as the sender.
Formally,
we say that the hypergraph $G$ is \emph{undirected}
if
\[
  \exists (u, S) \in E
  \quad \iff \quad
  \forall v \in S, \, \exists (v, S \cup \set{u} - v) \in E.
\]
We say that the hypergraph $G$ is \emph{bidirectional}
if the underlying simple graph $\underlying{G}$ is undirected.
This corresponds to the case where for every pair of nodes $u, v \in V$,
\[
  \exists e \in E \suchthat u = \head(e),\, v \in \tail(e)
  \quad \iff \quad
  \exists e' \in E \suchthat v = \head(e'),\, u \in \tail(e').
\]
Note that if a hypergraph $G$ is undirected,
then it is bidirectional.
However, the converse is not true.
For example, the local broadcast model on undirected graphs
is a special case of bidirectional hypergraphs
but not undirected hypergraphs.

The classical point-to-point communication model corresponds to the
case where each directed hyperedge has a single tail node.
This means that the hypergraph $G$ is essentially the same as
the underlying simple graph $\underlying{G}$.
So each edge $(u, v)$ in the graph $\underlying{G}$
represents a point-to-point channel where the messages
sent by node $u$ to node $v$ are private between $u$ and $v$.
Under the point-to-point communication model,
it is well known that $n \ge 3f + 1$
\cite{Fischer1986,Lamport:1982:BGP:357172.357176, Pease:1980:RAP:322186.322188}
and node connectivity at least $2f + 1$
\cite{DOLEV198214, Fischer1986}
are both necessary and sufficient for consensus
in arbitrary \emph{undirected} graphs.
The following theorem states that if $G$ is bidirectional
and has only point-to-point links,
i.e., each hyperedge has a single tail node,
then \conditionNC{} reduces to $n \ge 3f+1$ and
node connectivity $\ge 2f + 1$.

\begin{theorem} \label{thm reduction p2p}
  A bidirectional hypergraph $G$,
  such that each hyperedge has exactly one tail node,
  satisfies \conditionNC{}
  if and only if
  \begin{enumerate}[label=\arabic*)]
    \item $n \ge 3f + 1$, and
    \item the underlying undirected graph $\underlying{G}$
          has node connectivity at least $2f + 1$.
  \end{enumerate}
\end{theorem}

In \sectionRef{section proof p2p},
we show a more general result
(\theoremRef{thm reduction p2p directed})
when $G$
is not necessarily bidirectional,
but each hyperedge has exactly one tail node.
This corresponds to the point-to-point communication model
on arbitrary \emph{directed} graphs
\cite{LewisByzantineDirected,LewisByzantineDirectedARXIV}.
\theoremRef{thm reduction p2p} follows as a corollary.

The local broadcast model corresponds to the
other extreme
where each node $u$ in $G$ has exactly one hyperedge in $\incidentedges_G(u)$,
so that the messages transmitted by $u$
are received identically and correctly by all out-neighbors of $u$.
Under the local broadcast model,
our earlier work \cite{khan2019undirectedPODC}
shows that node degree at least $2f$
and connectivity at least $\floor{3f/2} + 1$
are both necessary and sufficient
for consensus in arbitrary \emph{undirected} graphs.
The following theorem states that
if $G$ is bidirectional and has only local broadcast channels,
i.e.,
each node
is a head node of a single hyperedge,
then \conditionNC{} reduces to
minimum node degree $\ge 2f$ and node connectivity $\ge \floor{3f/2} + 1$.

\begin{theorem} \label{thm reduction local broadcast}
  A bidirectional hypergraph $G$,
  such that each node is a head node of exactly one hyperedge,
  satisfies \conditionNC{}
  if and only if
  for the underlying undirected graph $\underlying{G}$
  \begin{enumerate}[label=\arabic*)]
    \item each node in $\underlying{G}$ has degree at least $2f$, and
    \item $\underlying{G}$ has node connectivity at least $\floor{3f/2} + 1$.
  \end{enumerate}
\end{theorem}

In \sectionRef{section proof local broadcast},
we show a more general result
(\theoremRef{thm reduction local broadcast directed})
when $G$ is not necessarily bidrectional,
but each node is a head node of exactly one hyperedge.
This corresponds to the local broadcast model
on arbitrary \emph{directed} graphs \cite{khan2020directedOPODIS}.
\theoremRef{thm reduction local broadcast} follows as a corollary.

The last model we consider in this section is the undirected hypergraph model.
Ravikant el. al. \cite{Ravikant10.1007/978-3-540-30186-8_32}
obtained tight conditions for this model.
Recall that a hypergraph $G = (V, E)$ is undirected if,
for every hyperedge $(u, S) \in E$ and tail node $v \in S$,
there exists a hyperedge $(v, (S - v) \cup \set{u})$.
For simplicity,
an undirected hyperedge $e$ can be viewed as a subset of nodes $e \subseteq V$,
representing $\abs{e}$ \emph{directed} hyperedges.
$e$ is called an $\abs{e}$-hyperedge.
Each hyperedge is effectively a local multicast channel
where \emph{any} node $u \in e$ can send a message,
which will be received identically and correctly by all nodes in $e - u$.

The tight characterization of undirected hypergraphs for consensus was given
by Ravikant et al. \cite{Ravikant10.1007/978-3-540-30186-8_32}.
We state this in \theoremRef{thm 2 3 hypergraph} below.
Observe that this is different from Theorem 1 in
\cite{Ravikant10.1007/978-3-540-30186-8_32}.
This is because we found a bug in the proof of Lemma 3 in
\cite{Ravikant10.1007/978-3-540-30186-8_32}
which is, in fact, not true, as documented in
\appendixRef{appendix hypergraphs misc}.
However,
\theoremRef{thm 2 3 hypergraph} still follows from the work in
\cite{Ravikant10.1007/978-3-540-30186-8_32}.
We observe that while this was presented as a tight characterization
for $(2,3)$-hypergraphs,\footnote{
  An undirected hypergraph $G$ is a $(2,3)$-hypergraph
  if each hyperedge is either a $2$-hyperedge or a $3$-hyperedge.
}
it also holds for general undirected hypergraphs.

\begin{theorem}
  [Fixed version of Theorem 1 in \cite{Ravikant10.1007/978-3-540-30186-8_32}]
  \label{thm 2 3 hypergraph}
  Byzantine consensus tolerating at most $f$ faulty nodes
  is achievable on an undirected hypergraph $G = (V, E)$ if and only if
  $G$ satisfies each of the following:
  \begin{enumerate}[label=\arabic*)]
    \item
    $n \ge 2f + 1$,

    \item
    the underlying simple graph $\underlying{G}$ is
    either a complete graph or is $(2f+1)$-connected,

    \item
    for every $V_1, V_2, V_3 \subseteq V$ such that
    $V_1 \cup V_2 \cup V_3 = V$ and $\abs{V_1} = \abs{V_2} = \abs{V_3} = f$,
    there exist three nodes
    \begin{enumerate}[label=(\roman*)]
      \item $u \in V_1 - (V_2 \cup V_3)$,
      \item $v \in V_2 - (V_1 \cup V_3)$, and
      \item $w \in V_3 - (V_1 \cup V_2)$,
    \end{enumerate}
    such that there is an undirected hyperedge in $G$
    that contains $u$, $v$, and $w$.
  \end{enumerate}
\end{theorem}

The following theorem states that if $G$ is an undirected hypergraph,
then \conditionNC{} reduces to the conditions in
\theoremRef{thm 2 3 hypergraph}.

\begin{theorem} \label{theorem hypergraph to undirected}
  An undirected hypergraph $G$ satisfies \conditionNC{}
  if and only if
  $G$ satisfies each of the following:
  \begin{enumerate}[label=\arabic*)]
    \item
    $n \ge 2f + 1$,

    \item
    the underlying simple graph $\underlying{G}$ is
    either a complete graph or is $(2f+1)$-connected,

    \item
    for every $V_1, V_2, V_3 \subseteq V$ such that
    $V_1 \cup V_2 \cup V_3 = V$ and $\abs{V_1} = \abs{V_2} = \abs{V_3} = f$,
    there exist three nodes
    \begin{enumerate}[label=(\roman*)]
      \item $u \in V_1 - (V_2 \cup V_3)$,
      \item $v \in V_2 - (V_1 \cup V_3)$, and
      \item $w \in V_3 - (V_1 \cup V_2)$,
    \end{enumerate}
    such that there is an undirected hyperedge in $G$
    that contains $u$, $v$, and $w$.
  \end{enumerate}
\end{theorem}

The formal proof of the theorem is given in
\sectionRef{section proof hypergraph}.

\section{Application to New Models} \label{section new models}
As mentioned in \sectionRef{section introduction},
the local multicast model also encompasses
some other network models of practical interest
that, to the best of our knowledge,
have not been considered before in the literature.
Suppose the $n$ nodes are connected via a local multicast network
represented as a directed hypergraph $G_1$.
For example,
network connectivity in $G_1$ can be via point-to-point links
or via wireless channels modelled as local broadcast.
Additionally, the $n$ nodes are connected via another local multicast network
represented as a directed hypergraph $G_2$.
For example, $G_2$ may correspond to a wireless network
with different frequencies and/or technologies.
The complete system, where nodes can communicate on
channels in $G_1$ as well as on channels in $G_2$,
can also be characterized by the local multicast model.
We omit details for brevity,
but this corresponds to the natural union of $G_1$ and $G_2$,
with each node now having access to its multicast channels in $G_1$
as well as its multicast channels in $G_2$.
 
\section[Necessity of Condition \nameOfConditionHNC{}]
{Necessity of \ConditionNC{}} \label{section necessity}
Intuitively, consider a set $F \subseteq V$ of size at most $f$,
such that $G$ violates \conditionNC{} with parameter $F$.
With $F$ as a candidate faulty set,
the splitting of nodes in $F$ captures possible equivocation
by nodes in $F$:
a faulty node can behave as if it has input $0$ on some of its
hyperedges and behave as if it has input $1$ on the other
hyperedges.
Now consider the execution where non-faulty nodes in $L$ have input $0$.
Since $R \cup C \notadjacent{G'} L - F'$,
nodes in $L - F'$ can not distinguish between $F$ and
its neighbors in $R \cup C$, i.e.,
$\inneighborhood{R \cup C}{L - F'}{G'}$ as the set of faulty nodes.
So non-faulty nodes in $L$ are stuck with outputting $0$ in this case.
Similarly, if non-faulty nodes in $R$ have input $1$,
then they have no choice but to output $1$,
creating the desired contradiction.

A formal necessity proof is given in
\sectionRef{section proof necessity multiple}.
It follows the standard state machine based approach
\cite{Attiya:2004:DCF:983102,DOLEV198214,Fischer1986},
similar to \cite{khan2020directedOPODIS,LewisByzantineDirected}.
Suppose there exists a set $F \subseteq V$, of size at most $f$,
such that $G$ does not satisfy \conditionNC{} with parameter $F$,
but there exists an algorithm $\mathcal A$ that solves consensus on $G$.
Algorithm $\mathcal A$ outlines a procedure $\mathcal A_u$
for each node $u$ that describes $u$'s state transitions,
as well as messages transmitted on each channel of $u$ in each round.
Now there exists a hypergraph $G' \in \graphSplitFSet_F(G)$ and
a partition of $V'$ that does not satisfy the requirements
of \conditionNC{}.
To create the required contradiction,
we work with an algorithm for $G'$ instead of $\mathcal A$.
To see why this works,
observe that an algorithm $\mathcal A$ on hypergraph $G$
can be adapted to create
an algorithm $\mathcal A'$
for a hypergraph $G' \in \graphSplitFSet_F(G)$
as follows.
Consider a round $i$ in the algorithm $\mathcal A$.
We specify the steps for each node in $G'$ in round $i$
for the algorithm $\mathcal A'$.
Each node $v \in V' \cap V$
that was not split runs $\mathcal A_v$ as specified for round $i$.
For a node $v \in V' - V$ that was split into $v^0, v^1 \in V'$,
both $v^0$ and $v^1$ run $\mathcal A_v$ for round $i$
with the following modification.
Consider a hyperedge $e \in \incidentedges_G(v)$.
Let $e' \in \incidentedges_{G'}(v^0)$ (resp. $e' \in \incidentedges_{G'}(v^1)$)
be the corresponding hyperedge in $G'$.
If the algorithm $\mathcal A_v$ wants to transmit
a message on $e$,
then $v^0$ (resp. $v^1$)
sends the message on $e'$,
while $v^1$ (resp. $v^0$) ignores this message transmission.
Observe that, for any node $u \in \tail(e')$,
$u$ receives messages on the hyperedge from exactly one of $v^0$ and $v^1$.
Furthermore, by construction of $G'$,
each node $v \in V'$ receives all messages
needed to run the corresponding next steps in the algorithm $\mathcal A'$.

Now, $\mathcal{A'}$ might not solve consensus on $G'$,
or may not even terminate.
In the following lemma,
we show that as long as care is taken with regards to
which nodes are allowed to be faulty
in $G'$
and the input of the split nodes,
$\mathcal{A'}$ indeed solves consensus in $G'$.
So for necessity,
it is enough to show that no algorithm exists
for a hypergraph $G' \in \graphSplitFSet_F(G)$,
under the two identified conditions.
We use this in the formal necessity proof in
\sectionRef{section proof necessity multiple}.

\begin{lemma} \label{lemma hypergraphs algorithm G implies G'}
  For a directed hypergraph $G = (V, E)$,
  a set $F \subseteq V$ of size at most $f$,
  and a hypergraph $G' \in \graphSplitFSet_F(G)$,
  if there exists a Byzantine consensus algorithm $\mathcal{A}$
  on $G$ tolerating at most $f$ faulty nodes,
  then there exists an algorithm $\mathcal{A}'$
  on $G' = (V', E')$
  that solves the Byzantine consensus problem under the following conditions.
  \begin{enumerate}[label=\arabic*)]
    \item The faulty nodes in $G'$ correspond to at most $f$ nodes in $G$.
    \item For each node $v \in F - V'$ that was split into $v^0, v^1 \in V'$,
    either
    \begin{enumerate}[label=(\roman*)]
      \item 
      both ${v^0}$ and ${v^1}$ have the same input, or
      \item
      at least one of ${v^0}$ and ${v^1}$ is faulty.
    \end{enumerate}
  \end{enumerate}
\end{lemma}

\begin{proof}
  Suppose there exists an arbitrary directed hypergraph $G = (V, E)$
  such that there is a consensus algorithm $\mathcal{A}$ for $G$
  tolerating $\le f$ Byzantine faults.
  Consider any set $F \subseteq V$ of size at most $f$
  and a hypergraph $G' \in \graphSplitFSet_F(G)$.
  Construct an algorithm $\mathcal{A'}$ from $\mathcal{A}$
  as described in the text preceding the lemma.
  We show that $\mathcal{A'}$ solves the Byzantine consensus problem under the
  conditions in the lemma statement.

  Consider an execution $\mathcal{E}'$ of $\mathcal{A'}$ on $G'$
  under the two conditions in the lemma statement.
  Without loss of generality,
  we assume that for every node $v \in F - V'$
  that was split into $v^0, v^1 \in V'$,
  either both $v^0, v^1$ are non-faulty in $\mathcal{E}'$
  or both are faulty.
  Observe that the faulty nodes in $\mathcal{E}'$ still
  correspond to at most $f$ nodes in $G$.
  Now for each node $v \in F - V'$ that was split into $v^0, v^1 \in V'$,
  either
  \begin{enumerate}[label=\roman*)]
    \item 
    both ${v^0}$ and ${v^1}$ have the same input, or
    \item
    both ${v^0}$ and ${v^1}$ are faulty.
  \end{enumerate}
  It follows,\footnote{
    Recall from the split operation (\sectionRef{section notation})
    that for every node $v \in V'$ in the hypergraph $G'$,
    each hyperedge in $\incidentedges_{G'}(v)$ corresponds to
    a single hyperedge in $G$.
  }
  by construction of $\mathcal{A'}$,
  that the behavior of each node $v' \in V'$
  on a hyperedge $e' \in \incidentedges_{G'}(v')$
  in any round of $\mathcal{E}'$
  is modelled by
  the behavior of the corresponding node $v \in V$
  on the corresponding hyperedge $e \in \incidentedges_{G}(v)$
  in the corresponding round of $\mathcal{E}$.
  
  Since $\mathcal{A}$ solves consensus on $G$ while
  tolerating $\le f$ faulty nodes,
  so all non-faulty nodes in $\mathcal{E}$
  terminate in finite time, agreeing on an input of some non-faulty node.
  Recall that the behavior of each node $v' \in V'$ in execution $\mathcal{E}'$
  is modelled
  by the behavior of the corresponding node $v \in V$ in execution $\mathcal E$.
  Therfore,
  as required,
  all non-faulty nodes in $\mathcal{E}'$
  also terminate in finite time, agreeing on an input of some non-faulty node.
\end{proof} 
\section{Algorithm for Directed Hypergraphs} \label{section sufficiency}
To prove the sufficiency portion of \theoremRef{theorem hypergraphs main},
we work with a different network condition,
which we will show to be equivalent to \conditionNC{}.
We first introduce some notation that is used in the algorithm.
For a set of nodes $U \subseteq V$,
we use $G[U] = (V_U, E_U)$
to denote the sub-hypergraph induced by the nodes in $U$,
i.e.,
\begin{align*}
  V_U &:= U, \\
  E_U &:= \set{
    (u, S) \mid \exists e \in E \suchthat
      \head(e) = u \in U, \,
      \tail(e) \cap U = S, \,
      S \ne \emptyset
  }.
\end{align*}
We use the shorthand $G - U$ to denote the sub-hypergraph $G[V - U]$.
For an additional set of hyperedges $D \subseteq E$,
we use $G[U, D] = (V_{U,D}, E_{U,D})$ to denote the sub-hypergraph induced by the nodes in $U$
and the hyperedges in $D$, i.e.,
\begin{align*}
  V_{U,D} &:= U \cup \set{u \mid \exists e \in D \suchthat \text{
    $u = \head(e)$ or $u \in \tail(e)$
  }}, \\
  E_{U,D} &:= D \cup E_U.
\end{align*}
Observe that if $u \in V_{U,D} - U$,
then $\incidentedges_{G[U, D]}(u) \subseteq D$.

\paragraph{Paths in Hypergraph \boldmath$G$:}
We use the following notations for paths.
A \emph{path} in a hypergraph $G$
is an alternating sequence of distinct nodes and hyperedges,
starting and ending at two distinct nodes,
such that
\begin{itemize}
  \item if node $u$ immediately precedes a hyperedge $e$ in the sequence,
  then $u = \head(e)$ is the head of $e$, and
  \item if a hyperedge $e$ immediately precedes a node $u$ in the sequence,
  then $u \in \tail(e)$ is a tail of $e$.
\end{itemize}
For a path $P = u_1, e_1, u_2, e_2, \dots, e_{k-1}, u_k$,
we say that $P$ \emph{passes through} $u_1, \dots, u_k$.

\begin{observation} \label{observation hypergraphs paths}
  For a hypergraph $G$ with the underlying graph $\underlying{G}$,
  there exists a path in $G$
  that passes through some nodes $u_1, \dots, u_k$
  if and only if
  there exists a path in $\underlying{G}$
  that passes through $u_1, \dots, u_k$.
  A path in $G$ corresponds to a unique path in $\underlying{G}$
  that passes through the same nodes.
  But a path in $\underlying{G}$
  can possibly correspond to multiple paths in $G$
  that pass through the same nodes.
\end{observation}

\begin{itemize}
  \item {\emph{$uv$-paths:}}
  For two nodes $u,v\in V$,
  a $uv$-path $P_{uv}$ is a path from $u$ to $v$.
  $u$ is called the \emph{source} and $v$ the \emph{terminal} of $P_{uv}$.
  Any other node in $P_{uv}$ is called an \emph{internal} node of $P_{uv}$.
  Two $uv$-paths are \emph{node-disjoint}
  if they do not share a common internal node.
  
  \item {\emph{$Uv$-paths:}}
  For a set $U \subset V$ and a node $v \not \in U$,
  a $Uv$-path is a $uv$-path for some node $u \in U$.
  All $Uv$-paths have $v$ as terminal.
  Two $Uv$-paths are \emph{node-disjoint}
  if they do not have any nodes in common except the terminal node $v$.
  In particular, two node-disjoint $Uv$-paths have different source nodes.
  By definition,
  the number of disjoint $Uv$-paths is upper bounded by the size of the set $U$.
  Note the difference in definition between
  node-disjoint $uv$-paths and
  node-disjoint $Uv$-paths when $U = \set{u}$ is a singleton set.
  The former requires only internal nodes to be different,
  while the latter needs to have different source nodes as well. 
  For the former, there can be more than one such node-disjoint path,
  while for the latter, there is at most one.
  
  \item {\emph{Propagate:}}
  For two node sets $A, B \subseteq V$, we use $A \propagate{G} B$
  (read as $A$ ``propagates'' to $B$ in $G$)
  to denote that either
  \begin{enumerate}[label=(\roman*)]
    \item 
    $B = \emptyset$, or

    \item
    for every $v \in B$,
    there exist at least $f+1$ node-disjoint $Av$-paths
    in the hypergraph $G$.
  \end{enumerate}
  Note that the subscript is important.
  For example,
  if $B \ne \emptyset$, then
  for a set $X \subseteq V$ that is disjoint from both $A$ and $B$,
  $A \propagate{G - X} B$
  requires that for every $v \in B$,
  there exist at least $f+1$ node-disjoint $Av$-paths
  in $G$
  \emph{that do not contain any nodes from $X$}.
\end{itemize}

\paragraph{Directed Decomposition in Hypergraphs:}
A directed hypergraph $G$ is \emph{strongly connected}
if, for every pair of nodes $u, v$,
there is a $uv$-path as well as a $vu-$path in $G$.
A strongly connected sub-hypergraph of $G$ is called
a \emph{component} of $G$.
A \emph{directed decomposition} of $G$
partitions $G$ into $H_1, \dots, H_k$, with $k > 0$,
such that each $H_i$ is a maximal component of $G$.
A maximal component $H_i$ that has no in-neighbors,
i.e., $\inneighborhood{V - H_i}{H_i}{G} = \emptyset$,
is called a \emph{source component} of the decomposition.
In any directed decomposition of a hypergraph $G$,
there always exists at least one source component.

We now give a different network condition
which is equivalent to \conditionNC{},
but will be useful for specifying an algorithm
for the local multicast model and proving its correctness.
Recall that we use $F'$ to denote the set of nodes in $G'$ corresponding
to nodes in $F$ in $G$.

\refstepcounter{condition} \label{condition hypergraphs SC}
\begin{definition}[Condition \nameOfConditionHSC{}]
  \label{definition hypergraphs SC}
  A hypergraph $G$ satisfies
  \emph{condition \nameOfConditionHSC{} with parameter $F$}
  if for every $G' \in \graphSplitFSet_F(G)$
  and every partition $(A, B)$ of $V'$,
  we have that either
  \begin{enumerate}[label=\arabic*)]
    \item $A \propagate{G' - B \cap F'} B - F'$, or
    \item $B \propagate{G' - A \cap F'} A - F'$.
  \end{enumerate}
  We say that $G$ satisfies \emph{condition \nameOfConditionHSC{}},
  if $G$ satisfies condition \nameOfConditionHSC{} with parameter $F$
  for every set $F \subseteq V$ of cardinality at most $f$.
\end{definition}

The following theorem states that \conditionNC{}
is equivalent to \conditionSC{}.
It was shown for simple graphs in \cite{khan2020directedOPODIS},
but based on \obsRef{observation hypergraphs paths},
can be extended to directed hypergraphs as well.

\begin{theorem}[\cite{khan2020directedOPODIS}]
  \label{theorem hypergraphs menger}
  A hypergraph $G$ satisfies \conditionNC{}
  if and only if $G$
  satisfies \conditionSC{}.
\end{theorem}

We show the sufficiency of
\conditionSC{} (and hence \conditionNC{}) constructively.
For the rest of this section,
we assume that $G$ satisfies \conditionSC{}.
We defer all proofs to \sectionRef{section proof sufficiency multiple}.
The proposed algorithm is given in \algoRef{algorithm hypergraphs}.
It draws inspiration from algorithms in
\cite{khan2019undirectedPODC, khan2020directedOPODIS, LewisByzantineDirected}.
Each node $v$ maintains a binary state variable $\gamma_v$,
which we call $v$'s $\gamma$ value.
Each node $v$ initializes $\gamma_v$ to be its input value.

\LinesNotNumbered
\begin{algorithm}[p]
  \linespread{1.25}\selectfont
  \DontPrintSemicolon
  \caption{Proposed algorithm for Byzantine consensus
  under the local multicast model:
  Steps performed by node $v$ are shown here.}
  \label{algorithm hypergraphs}
  \SetAlgoVlined
  \SetKwFor{For}{For}{do}{end}
{\em Initialization:} $\gamma_v :=$ input value of node $v$

  \For{each $F \subseteq V$ such that $\abs{F} \le f$}{
    \medskip
    \underline{\Step{a}:}
    Perform directed decomposition of $G - F$.
    Let $S$ be the unique source component of the decomposition
    (\lemmaRef{lemma hypergraphs source unique}).

    \medskip
    \underline{\Step{b}:}
    \lIf{$v \in S \cup \inneighborhood{F}{S}{G}$}{flood value $\gamma_v$.}

    \medskip
    \underline{\Step{c}:}
    \If{$v \in S$}{
      \medskip
      Create a hypergraph $G'_v$ by splitting all nodes in $F$ as follows.
      Set
\[
        F'
          := \set{ u^0 \mid u \in F } \cup \set{ u^1 \mid u \in F }
        \quad \text{and} \quad
        V(G'_v)
          := (V - F) \cup F'.
\]
The edges of $G'_v$ are as determined by the split operation,
      with the following choices:
      for each node $u \in F$ and hyperedge $e \in \incidentedges_G(u)$,
      identify a single $uv$-path $P_{uv}$ (if it exists)
      in $G[S, \set{e}]$.
      If $P_{uv}$ exists
      and $v$ received value $0$ from $u$ along $P_{uv}$ in \step{b},
      then assign $e$ to $u^0$.
      Else, assign $e$ to $u^1$.

      For each node $u \in S$,
      identify a single $uv$-path $P_{uv}$ in $G - F$
      (\lemmaRef{lemma algorithm nodes connected multiple}).
      Note that path $P_{vv}$ trivially exists
      ($P_{vv}$ contains only $v$).
      Initialize $Z_v$ and $N_v$ as follows,
\begin{align*}
        Z_v
          &:= \set[1]{u^0 \mid u \in \inneighborhood{F}{S}{G}} \cup
              \set{ u \in S \mid
                \text{$v$ received $0$
                  along $P_{uv}$ in \step{b}}},\\
        N_v &:= \set[1]{u^1 \mid u \in \inneighborhood{F}{S}{G}} \cup (S - Z_v).
      \end{align*}
      \vspace{-2em}
    }

    \medskip
    \underline{\Step{d}:}\\
    \lIf{$Z_v \propagate{G'_v - (N_v \cap F')} N_v - F'$}
    {set $A_v := Z_v$ and $B_v := N_v$}
    \lElse
    {set $A_v := N_v$ and $B_v := Z_v$}
    \medskip

    \If{$v \in B_v - F'$ and, in \step{b},
      $v$ received a value $b \in \set{0,1}$ identically along
      any $f+1$ node-disjoint $A_v v$-paths in
      $G'_v - (B_v \cap F')$}
    {\medskip$\gamma_v := b$.}
    
    \medskip
    \underline{\Step{e}:}
    \lIf{$v \in S$}{flood value $\gamma_v$.}

    \medskip
    \underline{\Step{f}:}
    \If{$v \in V - S - F$ and, in \step{e},
    $v$ received a value $b \in \set{0,1}$ identically along
    any $f+1$ node-disjoint $S v$-paths in $G - F$}
    {\medskip$\gamma_v := b$.}
  }

  \bigskip
  Output $\gamma_v$
\end{algorithm}

The nodes use ``flooding'' to communicate with the rest of the nodes.
We refer the reader to \cite{khan2019undirectedPODC, khan2020directedOPODIS}
for details about the flooding primitive.
Briefly, when a node $u$ wants to flood a binary value $b \in \set{0, 1}$,
it transmits $b$ to all of its neighbors,
who forward it to their neighbors,
and so forth.
If a node $u$ receives a message on a hyperedge $e$,
then $u$ appends the channel id of $e$
when fowarding the message to its neighbors.
This way a node $v$, on receiving a message,
can trace the path that the message has travelled to reach $v$.
By adding some simple sanity checks,
one can assume that even a faulty node $v$
does indeed transmit some value, when it is $v$'s turn to forward a message.
In at most $n$ synchronous rounds, the value $b$ will be ``flooded'' in $G$.
However, faulty nodes may tamper messages when forwarding,
so some nodes may receive a value $\bar{b} \ne b$
along paths that contain faulty nodes.

The algorithm proceeds in phases.
Every iteration of the main \texttt{for} loop (starting at line 2)
is a phase numbered $1, \dots, 2^f$.
Let $F^*$ denote the actual set of faulty nodes.
Each iteration of the for loop, i.e. phase $> 0$,
considers a candidate faulty set $F$.
In this iteration,
nodes attempt to reach consensus,
by updating their $\gamma$ state variables,
assuming the candidate set $F$ is indeed faulty.
Each iteration has six steps.

\begin{itemize}
  \item
  In \step{a}, each node $v$ performs a directed decomposition
  of $G - F$.
  This decomposition must have a unique source component:
  \begin{lemma}[Similar to Lemma 6 in \cite{khan2020directedOPODIS}]
    \label{lemma hypergraphs source unique}
    If a hypergraph $G$ satisfies
    \conditionSC{},
    then for any set $F$ of size $\le f$,
    the directed decomposition of $G - F$ has a unique source component.
  \end{lemma}
  We remind the reader that all proofs in this section are
  deferred to \sectionRef{section proof sufficiency multiple}.
  Each node $v$ identifies this unique source component $S$ of $G - F$.
  In steps \texttt{(b)}-\texttt{(d)},
  nodes in $S$ will attempt to reach consensus on a single value,
  and then propagate that to the remaining nodes
  in steps \texttt{(e)} and \texttt{(f)}.

  \item
  In \step{b}, each node $v \in S \cup \inneighborhood{F}{S}{G}$
  floods its $\gamma_v$ value.
  Nodes in $S$ may not be able to reach consensus by themselves,
  but they can pull in the nodes in
  $\inneighborhood{F}{S}{G}$ to help:
  \begin{lemma}[Similar to Lemma 7 in \cite{khan2020directedOPODIS}]
    \label{lemma hypergraphs source SC}
    For a hypergraph $G = (V, E)$,
    that satisfies \conditionSC{},
    and a set $F \subseteq V$ of size $\le f$,
    let $S$ be the unique source component in
    the directed decomposition of $G - F$.
    Then $G[S \cup \inneighborhood{F}{S}{G}]$ satisfies
    \conditionSC{} with parameter $\inneighborhood{F}{S}{G}$.
  \end{lemma}
  
  \item
  In \step{c},
  each node $v \in S$
  splits all nodes in the candidate faulty set $F$
  to construct a hypergraph $G'_v \in \graphSplitFSet_F(G)$.
  For the assignment of hyperedges in $G'_v$,
  consider a node $u \in F$ and a hyperedge $e \in \incidentedges_G(u)$.
  $v$ assigns $e$ to $v^0$ if there exists a $uv$-path $P_{uv}$ such that
  \begin{enumerate}[label=\arabic*)]
    \item
    $e$ is the first hyperedge on $P_{uv}$,

    \item
    the rest of $P_{uv}$ is contained entirely in $G[S]$, and

    \item
    $v$ received value $0$ from $u$ along $P_{uv}$ in \step{b}.
  \end{enumerate}
  Otherwise, $v$ assigns $e$ to $v^1$.

  Next,
  $v$ partitions nodes in $G'_v$ that
  correspond to nodes in $S \cup \inneighborhood{F}{S}{G}$
  in the original hypergraph $G$,
  into sets $Z_v$ and $N_v$,
  as follows.
  If $u \in \inneighborhood{F}{S}{G}$,
  then $v$ places $u^0$ in $Z_v$ and $u^1$ in $N_v$.
  If $u \in S$,
  then $v$ identifies a single $uv$-path $P_{uv}$ in $G[S]$.
  Such a path always exists
  since $S$ is strongly connected by construction.
If $v$ received 0 along $P_{uv}$ in \step{b},
  then $v$ places $u$ in $Z_v$.
  Otherwise $v$ places $u$ in $N_v$.
  For the purpose of \step{c},
  node $v$ is deemed to have received its own $\gamma_v$
  value along path $P_{vv}$, containing only node $v$,
  in \step{b}.

  The hypergraph $G'_v$, and sets $Z_v$ and $N_v$,
  are created in a manner so that
  \begin{enumerate}[label=\arabic*)]
    \item
    when $F \ne F^*$, nodes in $S$ may disagree on these constructions,
    i.e., it is possible that in this iteration,
    for two non-faulty nodes $u, w \in S$,
    we have either
    \[
      G'_u \ne G'_w, \qquad Z'_u \ne Z'_w, \qquad \text{or} \qquad N'_u \ne N'_w,
    \]
but
    
    \item
    when $F = F^*$, all nodes in $S$ agree on these constructions,
i.e.,
    in this iteration,
    for any two non-faulty nodes $u, w \in S$,
    we have
    \[
      G'_u = G'_w, \qquad Z'_u = Z'_w, \qquad \text{and} \qquad N'_u = N'_w.
    \]
\end{enumerate}
  
  \item
  In \step{d}, based on the estimates created in \step{c},
  a node $v \in S$ may update its $\gamma_v$ value.
  The update rules ensure that

  \begin{enumerate}[label=\arabic*)]
    \item in each iteration,
    for each non-faulty node $v \in S$,
    its $\gamma_v$ value at the end of this
    equals the $\gamma$ state value of some non-faulty node
    at the beginning of the iteration
    (\lemmaRef{lemma hypergraphs validity}).

    \medskip
    \item when $F = F^*$,
    all non-faulty nodes in $S$ have identical $\gamma$ values
    at the end of this step
    (\lemmaRef{lemma hypergraphs agreement}).
  \end{enumerate}

  \item
  In the iteration where $F = F^*$,
  by the end of \step{d},
  nodes in $S$ have reached consensus by adopting a single value
  in each of their $\gamma$ states.
  In steps \texttt{(e)} and \texttt{(f)},
  nodes in $S$ propagate the consensus value
  to the rest of the nodes,
  using the following property:
  \begin{lemma}[Similar to Lemma 10 in \cite{khan2020directedOPODIS}]
    \label{lemma hypergraphs source propagates}
    For a hypergraph $G = (V, E)$,
    that satisfies \conditionSC{},
    and a set $F \subseteq V$ of size $\le f$,
    let $S$ be the unique source component in
    the directed decomposition of $G - F$.
    Then $S \propagate{G - F} V - S - F$.
  \end{lemma}
\end{itemize}

At the end, after all iterations of the main for loop,
each output node $v$ outputs its $\gamma_v$ value.

The correctness of \algoRef{algorithm hypergraphs}
relies on the following two key lemmas,
which are proven in \sectionRef{section proof sufficiency multiple}
along with the 3 lemmas stated above.
Recall that
$G$ satisfies \conditionSC{} and
we use $F^*$ to denote the actual set of faulty nodes.

\begin{lemma} \label{lemma hypergraphs validity}
  For a non-faulty node $v \in V - F^*$,
  its state $\gamma_v$ at the end of any given phase
  of \algoRef{algorithm hypergraphs}
  equals the state of some non-faulty node at the start of that phase.
\end{lemma}

\begin{lemma} \label{lemma hypergraphs agreement}
  Consider a phase $> 0$ of \algoRef{algorithm hypergraphs}
  wherein $F = F^*$.
  At the end of this phase,
  every pair of non-faulty nodes $u, v \in V - F^*$ have identical state,
  i.e., $\gamma_u=\gamma_v$.
\end{lemma}

\lemmaRef{lemma hypergraphs validity} ensures validity,
i.e., that the output of each non-faulty node
is an input of some non-fautly node.
It also ensures that agreement among non-faulty nodes, once acheived,
is not lost.
\lemmaRef{lemma hypergraphs agreement} ensures
that agreement is reached in at least one phase of the algorithm.
These two lemmas imply correctness of \algoRef{algorithm hypergraphs}
as shown in \sectionRef{section proof sufficiency multiple}. 
\section{Conclusion}
In this paper,
we introduced the
local multicast model which,
to the best our knowledge,
has not been studied before in the literature.
The local multicast model corresponds to directed hypergraphs
and encompasses
the point-to-point, local broadcast, and undirected hypergraph
communication models,
as well as some new models which have not been considered before.
We identified a tight network condition
for Byzantine consensus under the local multicast model,
along the lines of \cite{khan2020directedOPODIS,LewisByzantineDirected},
and proved its necessity and sufficiency.
When the local multicast model
represents one of point-to-point, local broadcast, or undirected hypergraph
communication models,
we showed how the identified network condition
reduces to the known tight requirements for the corresponding case.

\bibliography{bib}

\appendix

\section{Reduction to Point-to-Point Channels}
\label{section proof p2p}
In this section,
we consider the case where each hyperedge in the hypergraph $G$
has exactly one tail node.
This corresponds to the classical point-to-point communication model
on arbitrary \emph{directed} graphs.
In this case, Tseng and Vaidya
\cite{LewisByzantineDirectedARXIV,LewisByzantineDirected}
showed that the following network condition,
which is similar to \conditionNC{},
is both sufficient and necessary.

\refstepcounter{condition} \label{condition p2p NC}
\begin{definition}
	[Condition \nameOfConditionPNC{}
	\cite{LewisByzantineDirected,LewisByzantineDirectedARXIV}]
	\label{definition p2p NC}
	A directed graph $G$ satisfies
	{\emph{condition \nameOfConditionPNC{} with parameter $F$}}
	if for every partition $(L, C, R)$ of $V - F$,
	we have that either
	\begin{enumerate}[label=\arabic*)]
    \item $R \cup C \adjacent{G} L$, or
    \item $L \cup C \adjacent{G} R$.
  \end{enumerate}
  We say that $G$ satisfies \emph{condition \nameOfConditionPNC{}},
  if $G$ satisfies condition \nameOfConditionPNC{} with parameter $F$
  for every set $F \subseteq V$ of cardinality at most $f$.
\end{definition}

When $G$ is undirected,
\conditionpNC{} reduces to 
$n \ge 3f + 1$ and node connectivity at least $2f + 1$
\cite{LewisByzantineDirectedARXIV,LewisByzantineDirected}.
Here, we show that \conditionNC{} reduces to \conditionpNC{}
when each hyperedge in the hypergraph $G$
has exactly one tail node.
\theoremRef{thm reduction p2p} follows as a corollary.

\begin{theorem} \label{thm reduction p2p directed}
  A directed hypergraph $G$, such that each hyperedge has exactly one tail node,
  satisfies \conditionNC{} if and only if
  the underlying directed graph $\underlying{G}$ satisfies \conditionpNC{}.
\end{theorem}
\begin{proof}
  Since $G$ and $\underlying{G}$ are essentially the same in this case,
  so we simply refer to $G$ for both in this proof.
  We show the contrapositive in both directions.
  First, consider a set $F \subseteq V$ of size at most $f$
  such that $G$ does not satisfy \conditionNC{} with parameter $F$.
  We will show that $G$ does not satisfy
  \conditionpNC{} with parameter $F$ either.
  Now, there exists a graph $G' \in \graphSplitFSet_F(G)$ and
  a partition $(L', C', R')$ of $V'$ such that,
  using $F'$ to denote the set of nodes in $G'$
  corresponding to nodes in $F$ in $G$,
  \begin{enumerate}[label=\arabic*)]
    \item $L' \cup C' \notadjacent{G'} R' - F'$, and
    \item $R' \cup C' \notadjacent{G'} L' - F'$.
  \end{enumerate}
  We create a partition $(L, C, R)$ of $V - F$ as follows:
  \begin{align*}
    L &:= L' - F',\\
    R &:= R' - F',\\
    C &:= C' - F'.
  \end{align*}
  Observe that $(L, C, R)$ is a partition of $V - F$.
  Furthermore, both $L$ and $R$ are non-empty.
  Now, we have
  \begin{align*}
    \abs{\inneighborhood{L \cup C}{R}{G}}
      &=  \abs{\inneighborhood{(L' \cup C') - F'}{R' - F'}{G}}
\\
      &= \abs{\inneighborhood{(L' \cup C') - F'}{R' - F'}{G'}}
\\
      &\le \abs{\inneighborhood{L' \cup C'}{R' - F'}{G'}} \\
      &\le f.
  \end{align*}
  Similarly, $\abs{\inneighborhood{R \cup C}{L}{G}} \le f$ as well.
  Therefore,
  \begin{enumerate}[label=\arabic*)]
    \item $L \cup C \notadjacent{G} R$, and
    \item $R \cup C \notadjacent{G} L$.
  \end{enumerate}
  So $G$ does not satisfy \conditionpNC{} with parameter $F$,
  as required.

  For the other direction,
  consider a set $F \subseteq V$ of size at most $f$
  such that $G$ does not satisfy \conditionpNC{} with parameter $F$.
  We will show that $G$ does not satisfy
  \conditionNC{} with parameter $F$ either.
  Now, there exists a partition $(L, C, R)$ of $V - F$ such that
  \begin{enumerate}[label=\arabic*)]
    \item $L \cup C \notadjacent{G} R$, and
    \item $R \cup C \notadjacent{G} L$.
  \end{enumerate}
  We create a graph $G' \in \graphSplitFSet_F(G)$
  by splitting all nodes in $F$,
  with the following choices:
  for each node $u \in F$ and an edge $(u, v)$,
  if $v \in L$, then add a hyperedge $(u^0, v)$ in $G'$;
  otherwise add $(u^1, v)$ in $G'$.
We create a partition $(L', C', R')$ of $V'$ as follows:
  \begin{align*}
    L' &:= L \cup \set{u^0 \mid u \in F},  \\
    R' &:= R \cup \set{u^1 \mid u \in F},  \\
    C' &:= C.
  \end{align*}
  Observe that, by construction,
  nodes in $R' - F' = R$ have no in-neighbors in
  $L' \cap F' = \set{u^0 \mid u \in F}$ and
  nodes in $L' - F' = L$ have no in-neighbors in
  $R' \cap F' = \set{u^1 \mid u \in F}$.
  So we have
  \begin{align*}
    \abs{\inneighborhood{L' \cup C'}{R' - F'}{G'}}
       &=  \abs{\inneighborhood{L \cup C}{R}{G}}
       \le  f,\\
\abs{\inneighborhood{R' \cup C'}{L' - F'}{G'}}
       &=  \abs{\inneighborhood{R \cup C}{L}{G}}
       \le  f.
  \end{align*}
  Therefore,
  \begin{enumerate}[label=\arabic*)]
    \item $L' \cup C' \notadjacent{G'} R' - F'$, and
    \item $R' \cup C' \notadjacent{G'} L' - F'$.
  \end{enumerate} 
  So $G$ does not satisfy \conditionNC{} with parameter $F$, as required.
\end{proof} 
\section{Reduction to Local Broadcast Model}
\label{section proof local broadcast}
In this section,
we consider the case where each node
is a head node of exactly one hyperedge.
This corresponds to the local broadcast model
on arbitrary \emph{directed} graphs.
In thise case,
Khan et. al. \cite{khan2020directedOPODIS}
showed that the following network condition,
which is similar to \conditionNC{},
is both sufficient and necessary.

\refstepcounter{condition} \label{condition local NC}
\begin{definition}[Condition \nameOfConditionBNC{}
  \cite{khan2020directedOPODIS,NaqviBroadcast}]
	\label{definition local NC}
	A directed graph $G$ satisfies
	{\emph{condition \nameOfConditionBNC{} with parameter $F$}}
	if for every partition $(L, C, R)$ of $V$,
	we have that either
	\begin{enumerate}[label=\arabic*)]
    \item ${R \cup C} \adjacent{G} L - F$, or
		\item ${L \cup C} \adjacent{G} R - F$.
  \end{enumerate}	
	We say that $G$ satisfies {\emph{condition \nameOfConditionBNC{}}},
	if $G$ satisfies condition \nameOfConditionBNC{} with parameter $F$
	for every set $F \subseteq V$ of cardinality at most $f$.
\end{definition}

When $G$ is undirected,
\conditionbNC{} reduces to
minimum node degree at least $2f$
and node connectivity at least $\floor{3f/2} + 1$.
Here, we show that \conditionNC{} reduces to \conditionbNC{}
when each node in the hypergraph $G$ is a head node of exactly one hyperedge.
\theoremRef{thm reduction local broadcast}
follows as a corollary.

\begin{theorem} \label{thm reduction local broadcast directed}
  A directed hypergraph $G$,
  such that each node is a head node of exactly one hyperedge,
  satisfies \conditionNC{} if and only if
  the underlying directed graph $\underlying{G}$
  satisfies \conditionbNC{}.
\end{theorem}
\begin{proof}
We show the contrapositive in both directions.
  First, consider a set $F \subseteq V$ of size at most $f$
  such that $G$ does not satisfy \conditionNC{} with parameter $F$.
  We will show that $\underlying{G}$ does not satisfy
  \conditionbNC{} with parameter $F$.
  Now, there exists a hypergraph $G' \in \graphSplitFSet_F(G)$ and
  a partition $(L, C, R)$ of $V'$ such that,
  using $F'$ to denote the set of nodes in $G'$
  corresponding to nodes in $F$ in $G$,
  \begin{enumerate}[label=\arabic*)]
    \item $L \cup C \notadjacent{G'} R - F'$, and
    \item $R \cup C \notadjacent{G'} L - F'$.
  \end{enumerate}
Consider a node $v \in F$ that was split into $v^0, v^1$ in $G'$.
  Since there is exactly one hyperedge in $\incidentedges_G(v)$,
  it follows that either $v^0$ or $v^1$ has degree $0$ in $G'$.
  So at least one of $v^0$ and $v^1$
  is neither in $\inneighborhood{L \cup C}{R - F'}{G'}$
  nor in $\inneighborhood{R \cup C}{L - F'}{G'}$.
  Therefore, WLOG, we can assume that $v$ was not split in $G'$.
  Thus $F' = F$, $G' = G$, and
  $(L, C, R)$ is a partition of $V$ such that
  \begin{enumerate}[label=\arabic*)]
    \item $L \cup C \notadjacent{G} R - F$, and
    \item $R \cup C \notadjacent{G} L - F$.
  \end{enumerate}
  Observe that
  \begin{align*}
    \inneighborhood{L \cup C}{R - F}{G} &= \inneighborhood{L \cup C}{R - F}{\underlying{G}},  \\
    \inneighborhood{R \cup C}{L - F}{G} &= \inneighborhood{R \cup C}{L - F}{\underlying{G}}.
  \end{align*}
  So $G$ does not satisfy \conditionbNC{} with parameter $F$,
  as required.

  For the other direction,
  consider a set $F \subseteq V$ of size at most $f$
  such that $\underlying{G}$ does not satisfy \conditionbNC{} with parameter $F$.
  We will show that $G$ does not satisfy
  \conditionNC{} with parameter $F$ either.
  Now, there exists a partition $(L, C, R)$ of $V$ such that
  \begin{enumerate}[label=\arabic*)]
    \item
    $L \cup C \notadjacent{\underlying{G}} R - F$ and so
    $L \cup C \notadjacent{{G}} R - F$, and
    \item
    $R \cup C \notadjacent{\underlying{G}} L - F$ and so
    $R \cup C \notadjacent{{G}} L - F$.
  \end{enumerate}
  Since $G \in \graphSplitFSet_F(G)$ and
  $(L, C, R)$ is a partition of $V$,
  so $G$ does not satisfy \conditionNC{} with parameter $F$, as required.
\end{proof} 
\section{Reduction to Undirected Hypergraphs}
\label{section proof hypergraph}
\begin{proof}[Proof of \theoremRef{theorem hypergraph to undirected}]
  Directly from Lemmas
  \ref{lem NC implies n > 2f},
  \ref{lem NC implies 2f+1 connectivity},
  \ref{lem NC implies 3 f hyperconn}, and
  \ref{lem hypergraph implies NC}
  below.
\end{proof}

We first show that if a hypergraph $G$ is undirected
and satisfies \conditionNC{},
then $G$ satisfies each of the conditions in
\theoremRef{theorem hypergraph to undirected}.

\begin{lemma} \label{lem NC implies n > 2f}
  If an undirected hypergraph $G$
  satisfies \conditionNC{},
  then $n \ge 2f + 1$.
\end{lemma}
\begin{proof}
  Consider an undirected hypergraph $G$.
  We show the contrapositive that if $n \le 2f$,
  then there exists $F \subseteq V$ of size at most $f$
  such that $G$ does not satisfy \conditionNC{} with parameter $F$.
  Let $F = \emptyset$.
  Observe that $\graphSplitFSet_F(G) = \set{G}$.
  Partition $V$ into $(L, R)$ such that
  $0 < \abs{L} \le f$ and
  $0 < \abs{R} \le f$.
  With $C = \emptyset$, $(L, C, R)$ is a partition of $V$.
  But, since $0 < \abs{L}, \abs{R} \le f$, we have that
  \begin{enumerate}[label=\arabic*)]
    \item $L \cup C = L \notadjacent{G} R = R - F$, and
    \item $R \cup C = R \notadjacent{G} L = L - F$,
  \end{enumerate}
  as required.
\end{proof}

\begin{lemma} \label{lem NC implies 2f+1 connectivity}
  If an undirected hypergraph $G$ satisfies \conditionNC{},
  then
  the underlying simple graph $\underlying{G}$ is
  either a complete graph or is $(2f+1)$-connected.
\end{lemma}
\begin{proof}
  We show the contrapositive that,
  for an undirected hypergraph $G$,
  if the underlying simple graph $\underlying{G}$
  is neither a complete graph nor is $(2f+1)$-connected,
  then $G$ does not satisfy \conditionNC{}.
If $n \le 2f$, then by \lemmaRef{lem NC implies n > 2f},
  we have that $G$ does not satisfy \conditionNC{}.
  So suppose that $n \ge 2f+1$.

  First, in each of the following two cases,
  we show that there exists a set $X$ of size at most $2f$
  that partitions $V - X$ into $(A, B)$ such that
  $\abs{A}, \abs{B} > 0$ and
  there is no undirected hyperedge in $G$
  that contains a node from both $A$ and $B$.
  \begin{enumerate}[label={Case \arabic*:},leftmargin=*]
    \item
    $n = 2f+1$.\\
    Since the underlying simple graph $\underlying{G}$ is not complete,
    there exist two nodes $u$ and $v$
    such that there is no undirected hyperedge containing both $u$ and $v$.
    Then,
    choosing $A := \set{u}$, $B := \set{v}$, and $X = V - A - B$
    satisfies the requirements above.

    \item
    $n > 2f+1$.\\
    Since the underlying simple graph $\underlying{G}$ is
    not $(2f+1)$-connected and $n > 2f+1$,
    there exists a set $X$ of size at most $2f$ that partitions $V - X$
    into $(A, B)$ such that $\abs{A}, \abs{B} > 0$
    and there is no undirected hyperedge in $G$
    that contains a node from both $A$ and $B$,
    as required.
  \end{enumerate}
  Partition $X$ into $(F, C)$ such that $\abs{F} \le f$ and $\abs{C} \le f$.
  Recall that there is no \emph{undirected} hyperedge that contains a node from
  both $A$ and $B$.
  It follows that,
for each node $z \in F$
  and \emph{directed} hyperedge $e \in \incidentedges_G(z)$,
  either $A \cap \tail(e) = \emptyset$ or $B \cap \tail(e) = \emptyset$.

  Now, we create a graph $G' \in \graphSplitFSet_F(G)$ by splitting
  all nodes in $F$, with the following choices:
  for each node $z \in F$ and a directed hyperedge $e \in \incidentedges_G(z)$,
  if $A \cap \tail(e) \ne \emptyset$, then assign $e$ to $z^0$;
  otherwise assign $e$ to $z^1$.
  Observe that $z^0$ does not have any hyperedge with tail nodes in $B$
  and $z^1$ does not have any hyperedge with tail nodes in $A$.
  Let
  \begin{align*}
    L &:= A \cup \set{z^0 \mid z \in F},  \\
    R &:= B \cup \set{z^1 \mid z \in F}.
  \end{align*}
Then $(L, C, R)$ is a partition of $V'$.
  We use $F'$ to denote the nodes in $G'$ corresponding to nodes in $F$ in $G$.
  Recall that $A = L - F'$ and $B = R - F'$ are both non-empty.
  Note that nodes in $L - F' = A$ (resp. $R - F' = B$)
  do not have any in-neighbors in $R$ (resp. $L$)
  and $\abs{C} \le f$.
  It follows that
  \begin{enumerate}[label=\arabic*)]
    \item $L \cup C \notadjacent{G'} A = R - F'$, and
    \item $R \cup C \notadjacent{G'} B = L - F'$.
  \end{enumerate}
  Thus, $G$ does not satisfy \conditionNC{}, as required.
\end{proof}

\begin{lemma} \label{lem NC implies 3 f hyperconn}
  If an undirected hypergraph $G$ satisfies \conditionNC{},
  then,
  for every $V_1, V_2, V_3 \subseteq V$ such that
  $V_1 \cup V_2 \cup V_3 = V$ and $\abs{V_1} = \abs{V_2} = \abs{V_3} = f$,
  there exist three nodes
  \begin{enumerate}[label=(\roman*)]
    \item $u \in V_1 - (V_2 \cup V_3)$,
    \item $v \in V_2 - (V_1 \cup V_3)$, and
    \item $w \in V_3 - (V_1 \cup V_2)$,
  \end{enumerate}
  such that there is an undirected hyperedge in $G$
  that contains $u$, $v$, and $w$.
\end{lemma}
\begin{proof}
  Consider an undirected hypergraph $G$.
  We show the contrapositive that
  if $G$ does not satisfy the condition in the lemma,
  then $G$ does not satisfy \conditionNC{}.
  Suppose that there exist
  $V_1, V_2, V_3 \subseteq V$ such that
  \begin{enumerate}[label=(\roman*)]
    \item 
    $V_1 \cup V_2 \cup V_3 = V$,

    \item
    $\abs{V_1} = \abs{V_2} = \abs{V_3} = f$,
    and

    \item
    no \emph{undirected} hyperedge crosses all the three sets
    $V_1 - (V_2 \cup V_3)$, $V_2 - (V_1 \cup V_3)$, and $V_3 - (V_1 \cup V_2)$.
  \end{enumerate}
  By (i) and (ii) above, we have $n \le 3f$.
  Furthermore, if $n \le 2f$,
  then we are done by \lemmaRef{lem NC implies n > 2f}.
  So for the rest of the proof, we assume that $2f < n \le 3f$.
  Let
  \begin{align*}
    V'_1 &:= V_1 - \del{V_2 \cup V_3},  \\
    V'_2 &:= V_2 - \del{V_1 \cup V_3},  \\
    V'_3 &:= V_3 - \del{V_1 \cup V_2}.
  \end{align*}
Observe that if either of the three sets $V'_1$, $V'_2$, and $V'_3$ is empty,
  then $\abs{V} \le 2f$, a contradiction.
  So each of $V'_1$, $V'_2$, and $V'_3$ is non-empty.
  Recall, from (iii) above, that
  there is no \emph{undirected} hyperedge crossing all the three sets
  $V'_1$, $V'_2$, and $V'_3$.
  It follows that,
  for each node $z \in V'_1$
  and \emph{directed} hyperedge $e \in \incidentedges_G(z)$,
  either $V'_2 \cap \tail(e) = \emptyset$
  or $V'_3 \cap \tail(e) = \emptyset$.
  
  Let $F := V_1$.
  We create a graph $G' \in \graphSplitFSet_F(G)$ by splitting
  all nodes in $V'_1 = F - (V_2 \cup V_3)$,
  with the following choices in the node split operation:
  for each node $z \in V'_1$
  and a directed hyperedge $e \in \incidentedges_G(z)$,
  if $V'_2 \cap \tail(e) \ne \emptyset$, then assign $e$ to $z^0$;
  otherwise assign $e$ to $z^1$.
  Observe that $z^0$ does not have any edge with tail nodes in $V'_3$
  and $z^1$ does not have any edge with tail nodes in $V'_2$.
  Let
  \begin{align*}
    L &:= \del{V_2 - V_3} \cup \set{z^0 \mid z \in V'_1}, \\
    R &:= \del{V_3 - V_2} \cup \set{z^1 \mid z \in V'_1}, \\
    C &:= V_2 \cap V_3.
  \end{align*}
  Note that $L$, $C$, and $R$ are disjoint.
  Furthermore, we have
  \begin{align*}
    L \cup C \cup R
      &=  \set{z^0, z^1 \mid z \in V'_1} \cup
          \del{V_2 - V_3} \cup \del{V_3 - V_2} \cup \del{V_2 \cap V_3} \\
      &=  \set{z^0, z^1 \mid z \in V'_1} \cup
      \del{V_2 \cup V_3} \\
      &=  (V - V'_1) \cup \set{z^0, z^1 \mid z \in V'_1}  \\
      &=  V'.
  \end{align*}
  Therefore, $(L, C, R)$ is a partition of
  $V'$.
We use $F'$ to denote the nodes in $G'$ corresponding to nodes in $F$ in $G$.
  Now $R - F' = V'_3$ and $L - F' = V'_2$ are both non-empty.
  Therefore,
  \begin{align*}
    \inneighborhood{L \cup C}{R - F'}{G'}
      &=  \inneighborhood{V_2 \cup \set[0]{z^0 \mid z \in F}}{V'_3}{G'}  \\
      &=  \inneighborhood{V_2}{V'_3}{G'}  \\
      &\subseteq  {V_2}.
  \end{align*}
  Since $\abs{V_2} = f$, so $L \cup C \notadjacent{G'} R - F'$.
  Similarly, $R \cup C \notadjacent{G'} L - F'$.
  Thus, $G$ does not satisfy \conditionNC{}, as required.
\end{proof}

We now show that if $G$ satisfies each of the three conditions in
\theoremRef{theorem hypergraph to undirected},
then $G$ satisfies \conditionNC{}.

\begin{lemma} \label{lem hypergraph implies NC}
  An undirected hypergraph $G$
  satisfies \conditionNC{} if $G$ satisfies each of the following:
  \begin{enumerate}[label=\arabic*)]
    \item
    $n \ge 2f + 1$,

    \item
    the underlying simple graph $\underlying{G}$ is
    either a complete graph or is $(2f+1)$-connected,

    \item
    for every $V_1, V_2, V_3 \subseteq V$ such that
    $V_1 \cup V_2 \cup V_3 = V$ and $\abs{V_1} = \abs{V_2} = \abs{V_3} = f$,
    there exist three nodes
    \begin{enumerate}[label=(\roman*)]
      \item $u \in V_1 - (V_2 \cup V_3)$,
      \item $v \in V_2 - (V_1 \cup V_3)$, and
      \item $w \in V_3 - (V_1 \cup V_2)$,
    \end{enumerate}
    such that there is an undirected hyperedge in $G$
    that contains $u$, $v$, and $w$.
  \end{enumerate}
\end{lemma}
\begin{proof}
  Consider an undirected hypergraph $G$.
  We show the contrapositive that
  if there exists a set $F \subseteq V$, of size at most $f$,
  such that $G$ does not satisfy \conditionNC{} with parameter $F$,
  then one of the conditions in the lemma statement is violated.
  Now, there exists a graph $G' \in \graphSplitFSet_F(G)$ and
  a partition $(L, C, R)$ of $V'$ such that,
  using $F'$ to denote the set of nodes in $G'$
  corresponding to nodes in $F$ in $G$,
  \begin{enumerate}[label=\arabic*)]
    \item $L \cup C \notadjacent{G'} R - F'$, and
    \item $R \cup C \notadjacent{G'} L - F'$.
  \end{enumerate}
  Note that this implies that both $R - F'$ and $L - F'$ are non-empty.
  There are the following cases to consider.
  \begin{enumerate}[label={Case \arabic*:},leftmargin=*]
    \item Either
    $(L \cup C) \not \subseteq F' \cup \inneighborhood{L \cup C}{R - F'}{G'}$
    or
    $(R \cup C) \not \subseteq F' \cup  \inneighborhood{R \cup C}{L - F'}{G'}$.
    \\
    Suppose that $(L \cup C) - F' - \inneighborhood{L \cup C}{R - F'}{G'}$
    is non-empty (the other case follows similarly).
    Let $B := R - F'$ (recall that $R - F'$ is non-empty).
    Let
    \begin{align*}
      A &:= (L \cup C) - F' - \inneighborhood{L \cup C}{R - F'}{G'},\\
      B &:= R - F',\\
      X' &:= F' \cup \inneighborhood{L \cup C}{R - F'}{G'}.
    \end{align*}
    Observe that both $A$ and $B$ are non-empty.
    Then,
    removing $X'$
    partitions $V' - X'$ into $(A, B)$
    such that there is no undirected hyperedge between $A$ and $B$ in $G'$.
    Note that $A \cup B \subseteq V \cap V'$.
    Let $X$ be the set of nodes in $G$ that correspond
    to nodes in $X'$ in $G'$.
    Then $X$ is a cut in the underlying simple graph $\underlying{G}$
    that partitions $V - X$ into $(A, B)$.
    We bound the size of $X$ as follows.
    By definition of $X$ and $X'$,
    \begin{align*}
      \abs{X}
        &= \abs{F \cup \del{\inneighborhood{L \cup C}{R - F'}{G'} - F'}}  \\
        &= \abs{F \cup \inneighborhood{(L \cup C) - F'}{R - F'}{G'}} \\
        &\le \abs{F} + \abs{\inneighborhood{(L \cup C) - F'}{R - F'}{G'}} \\
        &\le f + f  \\
        &= 2f.
\end{align*}
Therefore,
    since $A$ and $B$ are both non-empty,
    $X$ is a cut of size at most $2f$ in $\underlying{G}$.
    It follows that $\underlying{G}$ is
    neither a complete graph
    nor $(2f+1)$-connected.
    This violates the second condition in the lemma statement.

    \item
    $(L \cup C) \subseteq F' \cup \inneighborhood{L \cup C}{R - F'}{G'}$,
    $(R \cup C) \subseteq F' \cup \inneighborhood{R \cup C}{L - F'}{G'}$,
    and
    $\inneighborhood{C-F'}{R - F'}{G'} \ne \inneighborhood{C-F'}{L - F'}{G'}$.
    \\
    Without loss of generality assume that
    $\inneighborhood{C-F'}{R - F'}{G'} - \inneighborhood{C-F'}{L - F'}{G'}$
    is non-empty.
    Note that, by assumption of Case 2,
    \begin{align*}
      V'  &= L \cup R \cup C \\
          &= F' \cup
            \inneighborhood{L \cup C}{R - F'}{G'} \cup
            \inneighborhood{R \cup C}{L - F'}{G'}.
    \end{align*}
    Let
    \begin{align*}
      A &:= L - F',  \\
      B &:=
      \inneighborhood{C - F'}{R - F'}{G'} -
      \inneighborhood{C - F'}{L - F'}{G'},
      \\
      X' &:= F' \cup \inneighborhood{R \cup C}{L - F'}{G'}.
    \end{align*}
Observe that $A$ is non-empty since $L - F'$ is non-empty
    and $B$ is non-empty by assumption of Case 2.
    Then,
    removing
    $X'$ partitions $V' - X'$ into $(A, B)$.
    By construction of $B$,
    there is no \emph{undirected} hyperedge between $A$ and $B$.
    Note that $A \cup B \subseteq V \cap V'$.
    Let $X$ be the set of nodes in $G$ that correspond
    to nodes in $X'$ in $G'$.
    Then $X$ is a cut in the underlying simple graph $\underlying{G}$
    that partitions $V - X$ into $(A, B)$.
    We bound the size of $X$ as follows.
    By definition of $X$ and $X'$,
    \begin{align*}
      \abs{X}
        &= \abs{F \cup \del{\inneighborhood{R \cup C}{L - F'}{G'} - F'}}  \\
        &= \abs{F \cup \inneighborhood{(R \cup C) - F'}{L - F'}{G'}} \\
        &\le \abs{F} + \abs{\inneighborhood{(R \cup C) - F'}{L - F'}{G'}} \\
        &\le f + f  \\
        &= 2f.
\end{align*}
    Therefore,
    since $A$ and $B$ are both non-empty,
    $X$ is a cut of size at most $2f$ in $\underlying{G}$.
    It follows that $\underlying{G}$ is
    neither a complete graph
    nor $(2f+1)$-connected.
    This violates the second condition in the lemma statement.

    \item 
    $(L \cup C) \subseteq F' \cup \inneighborhood{L \cup C}{R - F'}{G'}$,
    $(R \cup C) \subseteq F' \cup \inneighborhood{R \cup C}{L - F'}{G'}$,
    and
    $\inneighborhood{C-F'}{R - F'}{G'} = \inneighborhood{C-F'}{L - F'}{G'}$.
    \\
First observe that, by assumption of Case 2,
    \begin{align*}
      V'  &= L \cup R \cup C \\
          &= F' \cup
            \inneighborhood{L \cup C}{R - F'}{G'} \cup
            \inneighborhood{R \cup C}{L - F'}{G'} \\
      \implies V &= F \cup
        \inneighborhood{(L \cup C) - F'}{R - F'}{G'} \cup
        \inneighborhood{(R \cup C) - F'}{L - F'}{G'}.
    \end{align*}
    This implies that $n = \abs{V} \le 3f$ as follows.
    \begin{align*}
      n &=  \abs{V} \\
        &=  \abs{ F \cup \inneighborhood{L \cup C - F'}{R - F'}{G'} \cup
              \inneighborhood{R \cup C - F'}{L - F'}{G'} } \\
        &\le  \abs{ F } + \abs{ \inneighborhood{L \cup C - F'}{R - F'}{G'}} +
              \abs{\inneighborhood{R \cup C - F'}{L - F'}{G'}} \\
        &\le  3f.
    \end{align*}
    Let
    \begin{align*}
      A' &:= \inneighborhood{L \cup C}{R - F'}{G'}, \\
      B' &:= \inneighborhood{R \cup C}{L - F'}{G'}.
    \end{align*}
    Let $A$ and $B$ denote the set of nodes in $G$
    corresponding to nodes in $A'$ and $B'$, respectively, in $G'$.
    Observe that $A \cup B \cup F = V$.
    Observe also that
    $\abs{A} \le \abs{A'} \le f$, $\abs{B} \le \abs{B'} \le f$, and
    $\abs{F} \le f$.
    If one of $A - (B \cup F)$, $B - (A \cup F)$, and $F - (A \cup B)$
    is empty, then $n \le 2f$,
    which violates the first condition in the lemma statement, and we are done.
    So assume that all three of the above sets is non-empty.
    By assumption of Case 3,
    \begin{align*}
      L - F'
        &\subseteq \inneighborhood{L - F'}{R - F'}{G'} \\
        &\subseteq L - F'.
    \end{align*}
    Therefore, $L - F' = \inneighborhood{L - F'}{R - F'}{G'}$,
    and we have
    \begin{align*}
      A - (B \cup F)
        &= (A - F) - (B - F)  \\
        &= (A' - F') - (B' - F')  \\
        &= \rlap{$\displaystyle \inneighborhood{(L \cup C) - F'}{R - F'}{G'}
            - \inneighborhood{(R \cup C) - F'}{L - F'}{G'}$}  \\
        &= \rlap{$\displaystyle \inneighborhood{L - F'}{R - F'}{G'}
        - \inneighborhood{R - F'}{L - F'}{G'}$}
        \\
&&\llap{since $\inneighborhood{C-F'}{R - F'}{G'}
                = \inneighborhood{C-F'}{L - F'}{G'}$,}\\
        &&\text{by assumption of Case 3}  \\
        &=  \inneighborhood{L - F'}{R - F'}{G'}
        &\text{since $L$ and $R$ are disjoint} \\
        &= L - F'
        &\text{since $L - F' = \inneighborhood{L - F'}{R - F'}{G'}$.}
\end{align*}
Similarly, $B - (A \cup F) = R - F'$.

    Now, we show that there is no \emph{undirected} hyperedge in $G$
    that crosses each of the 3 non-empty sets above.
    Consider any three nodes
    \begin{enumerate}[label=(\roman*)]
      \item $u \in A - (B \cup F) = L - F'$,
      \item $v \in B - (A \cup F) = R - F'$, and
      \item $z \in F - (A \cup B)$.
    \end{enumerate}
    If there is an \emph{undirected} hyperedge
    that contains all three of $u$, $v$, and $z$,
    then there is a \emph{directed} hyperedge $e \in \incidentedges_G(z)$
    such that $u, v \in \tail(e)$.
    We will create a contradiction with (iii) above, by showing that
    $z \in A \cup B$.
    Observe that $u, v \in V \cap V'$, i.e.,
    $u$ and $v$ were not split in $G'$.
    Let
    \begin{align*}
      z_e = \begin{cases}
        z &\qquad\text{if $z$ was not split in $G'$,}  \\
        z^0 &\qquad\text{if $z$ was split into $z^0, z^1$ in $G'$ and $e$ was assigned to $z^0$,}  \\
        z^1 &\qquad\text{if $z$ was split into $z^0, z^1$ in $G'$ and $e$ was assigned to $z^1$.}
      \end{cases}
    \end{align*}
    In each case,
    there is a \emph{directed} hyperedge $e'$ in $G'$,
    corresponding to $e$ in $G$,
    such that $z_e = \head(e')$ and $u, v \in \tail(e')$.
    Note that $z_e \in V' = L \cup C \cup R$
    and there are two cases to consider.
    \begin{itemize}[label={$z_e \in L \cup C$:},leftmargin=*]
      \item[$z_e \in L \cup C$:]
      then $z_e \in \inneighborhood{L \cup C}{R - F'}{G'} = A'$
      since $v \in \tail(e')$ and $v \in R - F'$.
      \item[$z_e \in R$:]
      then $z_e \in \inneighborhood{R \cup C}{L - F'}{G'} = B'$
      since $u \in \tail(e')$ and $u \in L - F'$.
    \end{itemize}
    In either case, $z_e \in A' \cup B'$.
    It follows that $z \in A \cup B$, a contradiction.
Therefore, there is no undirected hyperedge
    that contains all three of $u$, $v$, and $z$.

    Finally,
    we show that the third condition in the lemma statement
    is violated.
    Since $n > f$, we can find three sets
    $V_1 \supseteq A, V_2 \supseteq B, V_3 \supseteq F$ such that
    $\abs{V_1} = \abs{V_2} = \abs{V_3} = f$.
    Observe that $V_1 \cup V_2 \cup V_3 = A \cup B \cup F = V$,
    and so
    \begin{align*}
      V_1 - (V_2 \cup V_3) & \subseteq A - (B \cup F),  \\
      V_2 - (V_1 \cup V_3) & \subseteq B - (A \cup F),  \\
      V_3 - (V_1 \cup V_2) & \subseteq F - (A \cup B).
    \end{align*}
    Therefore,
    there is no undirected hyperedge in $G$ across the three
    (possibly empty) sets
    \begin{enumerate}[label=(\roman*)]
      \item $V_1 - (V_2 \cup V_3)$,
      \item $V_2 - (V_1 \cup V_3)$, and
      \item $V_3 - (V_1 \cup V_2)$.
    \end{enumerate}
    This violates the third condition in the lemma statement.
\end{enumerate}
  In all cases,
  we have that one of the conditions in the lemma statement is violated.
\end{proof} 
\section[Proof of Necessity of Condition \nameOfConditionHNC{}]
{Proof of Necessity of \ConditionNC{}}
\label{section proof necessity multiple}

In this section,
we show the necessity portion of \theoremRef{theorem hypergraphs main},
following the discussion in \sectionRef{section necessity}.

\begin{proof}
  [Proof of \theoremRef{theorem hypergraphs main} ($\Rightarrow$ direction)]
  Suppose for the sake of contradiction that there exists a set $F$,
  of cardinality at most $f$,
  such that $G$ does not satisfy \conditionNC{} with parameter $F$,
  but there exists an algorithm $\mathcal A$
  that solves Byzantine consensus on $G$.
  Then there is a hypergraph ${G'} \in \graphSplitFSet_F(G)$
  and a partition $(L, C, R)$ of $V'$ such that,
  using $F'$ to denote the set of nodes in $G'$
  corresponding to nodes in $F$ in $G$,
  \begin{enumerate}[label=\arabic*)]
    \item $L \cup C \notadjacent{{G'}} R - F'$, and
    \item $R \cup C \notadjacent{{G'}} L - F'$.
  \end{enumerate}
  Note that this implies that both $R - F'$ and $L - F'$ are non-empty.
  Consider the nodes in $C \cap F'$.
  By moving them from $C$ to $L$,
  the required condition stays violated.
  Therefore, without loss of generality, we assume that $C \cap F' = \emptyset$
  for the rest of the proof.
  As described in \sectionRef{section necessity},
  we work with the algorithm $\mathcal A'$ on $G'$
  that corresponds to $\mathcal A$,
  with appropriate inputs and faulty nodes to create the desired contradiction.
  
  We first create a directed hypergraph
  $\mathcal{G} = (\mathcal{V}, \mathcal{E})$
  to model the behavior
  of nodes in three different executions
  $E_1$, $E_2$, and $E_3$
  of algorithm $\mathcal A'$ on $G' = (V', E')$.
  We will describe these executions later.
  \figureRef{figure hypergraph necessity network} depicts the
  underlying simple graph
  $\underlying{\mathcal{G}} = (\underlying{\mathcal{V}}, \underlying{\mathcal{E}})$.
  Recall that for the set $F \subseteq V$,
  we use $F'$ to denote the corresponding nodes in $V'$.
  Let
  \begin{align*}
    L' &:= L - \inneighborhood{L}{R-F'}{{G'}},\\
    R' &:= R - \inneighborhood{R}{L-F'}{{G'}},\\
    C' &:= C - \del{\inneighborhood{C}{L-F'}{{G'}} \cup \inneighborhood{C}{R-F'}{{G'}}}.
  \end{align*}
  A node $u$ in ${G'}$ may have up to 3 copies in $\mathcal{G}$,
  denoted by $u_0, u_1, u_2$.
  If a node has a single copy in $\mathcal{G}$, then we omit the subscript.
  This notation extends to sets as well so that
  $C'_1, C'_2, C'_3$ denote the three copies of the nodes in $C'$.
  The nodes have the following number of copies in $\mathcal{G}$,
  as depicted in \figureRef{figure hypergraph necessity network}.
  \begin{itemize}
    \item
    Nodes in $C'$ have three copies.

    \item
    Nodes in $\inneighborhood{L}{R-F'}{{G'}}$,
    $\inneighborhood{R}{L-F'}{{G'}}$, and
    $\inneighborhood{C}{L-F'}{{G'}} \cap \inneighborhood{C}{R-F'}{{G'}}$
    have a single copy.

    \item
    All other nodes have two copies.
  \end{itemize}
  
  We describe the hyperedges of $\mathcal{G}$
  based on the hyperedges of $G'$ and
  the simple edges of $\underlying{\mathcal{G}}$
  depicted in \figureRef{figure hypergraph necessity network}.
  We use $v' \in \mathcal{V}$ to denote a copy of a node $v \in V'$.
  Consider a copy $u' \in \mathcal{V}$ of a node $u \in V'$.
  For a hyperedge $e = (u, S) \in \incidentedges_{G'}(u)$,
  let
  \[S' := \set{ v' \mid \text{$v \in S$ and $(u', v') \in \underlying{\mathcal{E}}$}}. \] 
  If $S' \ne \emptyset$, then $(u', S')$ is a hyperedge in $\mathcal{G}$.
  $\underlying{\mathcal{G}}$ has been constructed to ensure that
  for each edge $(u, v) \in \underlying{E'}$ of $\underlying{G'}$,
  each copy of $v$ has an edge
  from exactly one copy of $u$ in $\underlying{\mathcal{G}}$.
Hence, for each hyperedge $e = (u, S) \in E'$ of $G'$ such that $v \in S$,
  each copy of $v$ receives messages
  on exactly one hyperedge corresponding to $e$ in $\mathcal{G}$.
  However, there can be multiple copies of $v$
  that receive messages from a copy of $u$.

  The algorithm $\mathcal A'$ outlines a procedure $\mathcal A'_u$
  for each node $u \in V'$ that describes $u$'s state transitions,
  as well as messages transmitted to each neighbor $v$ of $u$ in each round.
  We create an algorithm for $\mathcal{G}$, corresponding to $\mathcal A'$,
  as follows.
  Consider a hyperedge $(u, S) \in E'$ in $G'$.
  Let $u'$ be a copy of $u$ in $\mathcal{G}$
  and let $(u', S')$ be the hyperedge in $\mathcal{G}$
  corresponding to the hyperedge $(u, S)$
  (using $S' = \emptyset$ for the case there is no such hyperedge).
  Then $u'$ runs the procedure $\mathcal A'_u$, with the following modification.
  When $\mathcal A'_u$ requires a message $m$
  to be sent on the hyperedge $(u, S)$,
  $u'$ sends the message $m$ on the hyperedge $(u', S')$ in $\mathcal{G}$.
  Recall that, by construction of $\mathcal{G}$,
  for any in-neighbor $v$ of node $u$ in $G'$,
  each copy of $u$ receives messages from exactly one copy of $v$
  in $\mathcal{G}$.
  So each copy of $u$ in $\mathcal{G}$
  can correctly run the procedure $\mathcal A'_u$.
  Observe that it is not guaranteed that
  the nodes will agree on the same value,
  or even if the algorithm will terminate.

  Consider an execution $\Sigma$ of the above algorithm on $\mathcal{G}$
  with the following inputs.
  All (copies of) nodes denoted with subscript $0$ have input $0$.
  All (copies of) nodes denoted with subscript $1$ have input $1$.
  $C'_2$ is the only set with subscript $2$, and has input $1$.
  For the single copy nodes,
  $\inneighborhood{L}{R-F'}{{G'}}$ has input $0$,
  while all others have input $1$.
  We show that with these inputs, the algorithm above does terminate,
  but the output of the nodes will help us in deriving the desired contradiction.
  We use the execution $\Sigma$ to model three executions
  $E_1$, $E_2$, and $E_3$
  of $\mathcal A'$ on the hypergraph ${G'}$.
  In each of the three executions,
  we ensure that the conditions of
  \lemmaRef{lemma hypergraphs algorithm G implies G'}
  are met so that
  $\mathcal A'$ solves consensus in finite time.
  $E_1$, $E_2$, and $E_3$ are as follows.

  \begin{enumerate}[label=$E_{\arabic*}:$,leftmargin=*]
    \item $\inneighborhood{R \cup C}{L - F'}{{G'}}$ is the set of faulty nodes
    in this execution.
    Recall that $\abs{\inneighborhood{R \cup C}{L - F'}{{G'}}} \le f$.
    All non-faulty nodes have input $0$.
    Observe that this satisfies the conditions of
    \lemmaRef{lemma hypergraphs algorithm G implies G'}
    so that $\mathcal A'$ solves consensus in finite time in this execution.
    \figureRef{figure hypergraph necessity exec 0}
    depicts the execution $E_1$.

Consider any arbitrary round in $E_1$.
    We describe the messages transmitted by faulty nodes in this round.
    If a faulty node $u$ has a single copy in $\mathcal{G}$,
    then, in $E_1$,
    $u$ transmits the same messages as the copy in execution $\Sigma$.
    If a faulty node $u \in V'$ has two copies
    $u_0$ and $u_1$ in $\mathcal{G}$,
    then, in $E_1$,
    $u$ transmits the same messages as the copy $u_0$ in execution $\Sigma$.
    \figureRef{figure hypergraph necessity exec 0}
    depicts how the behavior of each node, faulty or non-faulty,
    in $E_1$
    is modelled by the corresponding copy in $\Sigma$.
    Observe that each node in ${G'}$ is being modelled
    by exactly one copy in $\mathcal{G}$.
    Since $\mathcal A'$ solves Byzantine consensus on ${G'}$,
    so all non-faulty nodes decide on output $0$ (by validity) in finite time.
    In particular,
    all nodes in $L - F'$ in $E_1$ decide on output $0$.
    In $\Sigma$, these are modelled by copies in either
    $(L' - F')_0$ or $\inneighborhood{L-F'}{R-F'}{{G'}}$.
    Therefore, all nodes in
    $(L' - F')_0$ and $\inneighborhood{L-F'}{R-F'}{{G'}}$
    decide on output $0$ in $\Sigma$.

    \item $\inneighborhood{L \cup C}{R - F'}{{G'}}$ is the set of faulty nodes
    in this execution.
    Recall that $\abs{\inneighborhood{L \cup C}{R - F'}{{G'}}} \le f$.
    All non-faulty nodes have input $1$.
    Observe that this satisfies the conditions of
    \lemmaRef{lemma hypergraphs algorithm G implies G'}
    so that $\mathcal A'$ solves consensus in finite time in this execution.
    \figureRef{figure hypergraph necessity exec 1}
    depicts the execution $E_2$.

Consider any arbitrary round in $E_2$.
    We describe the messages transmitted by faulty nodes in this round.
    If a faulty node $u$ has a single copy in $\mathcal{G}$,
    then, in $E_2$,
    $u$ transmits the same messages as the copy in execution $\Sigma$.
    If a faulty node $u \in V'$ has two copies
    $u_0$ and $u_1$ in $\mathcal{G}$,
    then, in $E_2$,
    $u$ transmits the same messages as the copy $u_1$ in execution $\Sigma$.
    \figureRef{figure hypergraph necessity exec 1}
    depicts how the behavior of each node, faulty or non-faulty,
    in $E_2$
    is modelled by the corresponding copy in $\Sigma$.
    Observe that each node in ${G'}$ is being modelled
    by exactly one copy in $\mathcal{G}$.
    Since $\mathcal A'$ solves Byzantine consensus on ${G'}$,
    so all non-faulty nodes decide on output $1$ (by validity) in finite time.
    In particular,
    all nodes in $R - F'$ in $E_2$ decide on output $1$.
    In $\Sigma$, these are modelled by copies in either
    $(R' - F')_1$ or $\inneighborhood{R-F'}{L-F'}{{G'}}$.
    Therefore, all nodes in
    $(R' - F')_1$ and $\inneighborhood{R-F'}{L-F'}{{G'}}$
    decide on output $1$ in $\Sigma$.

    \item $F'$ is the set of faulty nodes.
    Recall that $C \cap F' = \emptyset$ and that some nodes in $F'$ in ${G'}$
    might have been split from original nodes in $F$ in $G$.
    However, $\abs{F} \le f$, i.e.
    the total number of corresponding faulty nodes in $G$ is at most $f$.
    \figureRef{figure hypergraph necessity exec 2}
    depicts the execution $E_3$.
    There are no split nodes outside of $F'$ in ${G'}$.
    Therefore,
    the conditions of
    \lemmaRef{lemma hypergraphs algorithm G implies G'}
    are satisfied and
    $\mathcal A'$ solves consensus in finite time in this execution.
All non-faulty nodes in the set $L - F'$ have input $0$.
    All non-faulty nodes in the set
    $\inneighborhood{C}{L - F'}{{G'}} - \inneighborhood{C}{R - F'}{{G'}}$
    also have input $0$.
    All the other non-faulty nodes have input $1$.
    Consider any arbitrary round in $E_3$.
    We describe the messages transmitted by faulty nodes in this round.
    If a faulty node $u \in F'$ has a single copy in $\mathcal{G}$,
    then, in $E_3$,
    $u$ transmits the same messages as the copy in execution $\Sigma$.
    If a faulty node $u \in F'$ has two copies $u_0$ and $u_1$ in $\mathcal{G}$,
    then, $u \in L' \cap F'$ (resp. $u \in R' \cap F'$) in $G'$.
    $u$ transmits the same messages as the copy $u_0$ (resp. $u_1$)
    in execution $\Sigma$.
    \figureRef{figure hypergraph necessity exec 2}
    depicts how the behavior of each node, faulty or non-faulty,
    in $E_3$
    is modelled by the corresponding copy in $\Sigma$.
    Observe that each node in $G'$ is being modelled
    by exactly one copy in $\mathcal{G}$,
    even if it comes from an original node in $F - F'$ in $G$ that was split.
    We show that the output of nodes in excecution $E_3$
    is not the same, thus deriving the contradiction.
  \end{enumerate}
  
  \def\SES{{"black","gray","black","gray"}}   \def\EFE{{"black","black","red","black"}}  \def\ESE{{"black","black","gray","black"}}  \def\SEE{{"black","gray","black","black"}}   \def\FEE{{"black","red","black","black"}}    \def\ESS{{"black","black","gray","gray"}}  \def\SFS{{"black","gray","red","gray"}}   \def\SEF{{"black","gray","black","red"}}   \def\ESF{{"black","black","gray","red"}}   \def\SSF{{"black","gray","gray","red"}}   \def\SSE{{"black","gray","gray","black"}}   \def\SFF{{"black","gray","red","red"}}    \def\SFE{{"black","gray","red","black"}}    \def\EFF{{"black","black","red","red"}}    \def\FEF{{"black","red","black","red"}}    \def\FFS{{"black","red","red","gray"}}   \def\FFE{{"black","red","red","black"}}  \def\FFF{{"black","red","red","red"}}  \def\FSS{{"black","red","gray","gray"}}  \def\FSF{{"black","red","gray","red"}}  \def\FSE{{"black","red","gray","black"}}  

\newcommand{\getColor}[1]{\pgfmathparse{#1[\theindex]}}
\newcommand{\withColor}[2]{
  \getColor{#1}
  #2
}

\begin{figure}[p]
  \centering
  \setcounter{index}{0}
\begin{tikzpicture}[scale=0.55, every node/.style={scale=0.55}]
  \withColor{\SES}{
    \node[draw, circle, minimum size=1cm,\pgfmathresult] at (0, 0) (L_1 cap F) {$(L' \cap F')_1$};
  }
  \withColor{\SES}{
    \node[draw, circle, minimum size=1cm,\pgfmathresult] at (0, -3) (L_1) {$(L' - F')_1$};
  }
  \withColor{\EFF}{
    \node[draw, ellipse, text width=1.75cm,\pgfmathresult, label=left:{\color{\pgfmathresult} $0$}]
      at (0, -6) (NL cap FR) {${\inneighborhood{}}_{G'} ({L \cap F'},$\\${R-F'})$};
  }
  \withColor{\EFE}{
    \node[draw, ellipse, text width=1.75cm,\pgfmathresult, label=left:{\color{\pgfmathresult} $0$}]
      at (0, -9) (NL minus FR) {${\inneighborhood{}}_{G'} ({L - F',}$\\${R-F'})$};
  }
  \withColor{\ESE}{
    \node[draw, circle, minimum size=1cm,\pgfmathresult]
      at (0, -12) (L_0) {$(L' - F')_0$};
  }
  \withColor{\ESF}{
    \node[draw, circle, minimum size=1cm,\pgfmathresult]
      at (0, -15) (L_0 cap F) {$(L' \cap F')_0$};
  }

  \withColor{\SEF}{
    \node[draw, circle, minimum size=1cm,\pgfmathresult]
      at (6, 0) (R_1 cap F) {$(R' \cap F')_1$};
  }
  \withColor{\SEE}{
    \node[draw, circle, minimum size=1cm,\pgfmathresult]
      at (6, -3) (R_1) {$(R' -F')_1$};
  }
  \withColor{\FEF}{
    \node[draw, ellipse, text width=1.75cm,\pgfmathresult, label=right:{\color{\pgfmathresult} $1$}]
      at (6, -6) (NR cap FL) {${\inneighborhood{}}_{G'}({R \cap F'},$ ${L-F'})$};
  }
  \withColor{\FEE}{
    \node[draw, ellipse, text width=1.75cm,\pgfmathresult, label=right:{\color{\pgfmathresult} $1$}]
      at (6, -9) (NR minus FL) {${\inneighborhood{}}_{G'}({R - F'},$ ${L-F'})$};
  }
  \withColor{\ESS}{
    \node[draw, circle, minimum size=1cm,\pgfmathresult]
      at (6, -12) (R_0) {$(R' -F')_0$};
  }
  \withColor{\ESS}{
    \node[draw, circle, minimum size=1cm,\pgfmathresult]
      at (6, -15) (R_0 cap F) {$(R' \cap F')_0$};
  }

  \withColor{\SFS}{\draw[\pgfmathresult,\arrow,out=135,in=225] (NL cap FR) to (L_1 cap F);}
  \withColor{\SFS}{\draw[\pgfmathresult,\arrow,out=180-15,in=180+15] (NL minus FR) to (L_1 cap F);}
  \withColor{\SFS}{\draw[\pgfmathresult,\arrow] (NL cap FR) to (L_1);}
  \withColor{\SFS}{\draw[\pgfmathresult,\arrow,out=180-30,in=180+30] (NL minus FR) to (L_1);}

  \withColor{\SES}{\draw[\pgfmathresult,-,out=45,in=-45] (NR cap FL) to (R_1 cap F);}
  \withColor{\SEF}{\draw[\pgfmathresult,-,out=15,in=-15] (NR minus FL) to (R_1 cap F);}
  \withColor{\SEF}{\draw[\pgfmathresult,-] (NR cap FL) to (R_1);}
  \withColor{\SEE}{\draw[\pgfmathresult,-,out=30,in=-30] (NR minus FL) to (R_1);}

  \withColor{\SES}{\draw[\pgfmathresult,-] (L_1 cap F) to (L_1);}
  \withColor{\SEF}{\draw[\pgfmathresult,-] (R_1 cap F) to (R_1);}

  \withColor{\SES}{\draw[\pgfmathresult,-] (L_1 cap F) to (R_1 cap F);}
  \withColor{\SES}{\draw[\pgfmathresult,-] (NR cap FL) to (L_1 cap F);}

  \withColor{\SFS}{\draw[\pgfmathresult,\arrow] (NL cap FR) to (R_1 cap F);}

  \withColor{\SFS}{\draw[\pgfmathresult,-] (NR cap FL) to (L_1);}

  \withColor{\SFF}{\draw[\pgfmathresult,\arrow] (NL cap FR) to (R_1);}

  \withColor{\SEF}{\draw[\pgfmathresult,-] (NR cap FL) to (NR minus FL);}
  \withColor{\ESF}{\draw[\pgfmathresult,-] (NL cap FR) to (NL minus FR);}
  \withColor{\FFS}{\draw[\pgfmathresult,-] (NL cap FR) to (NR cap FL);}
  \withColor{\FFE}{\draw[\pgfmathresult,-] (NL minus FR) to (NR minus FL);}
  \withColor{\FFF}{\draw[\pgfmathresult,-] (NL cap FR) to (NR minus FL);}
  \withColor{\FFF}{\draw[\pgfmathresult,-] (NR cap FL) to (NL minus FR);}

  \withColor{\FSS}{\draw[\pgfmathresult,\arrow,out=-15,in=15] (NR cap FL) to (R_0 cap F);}
  \withColor{\FSS}{\draw[\pgfmathresult,\arrow,out=-45,in=45] (NR minus FL) to (R_0 cap F);}
  \withColor{\FSS}{\draw[\pgfmathresult,\arrow,out=-30,in=30] (NR cap FL) to (R_0);}
  \withColor{\FSS}{\draw[\pgfmathresult,\arrow] (NR minus FL) to (R_0);}

  \withColor{\ESS}{\draw[\pgfmathresult,-,out=180+15,in=180-15] (NL cap FR) to (L_0 cap F);}
  \withColor{\ESF}{\draw[\pgfmathresult,-,out=180+45,in=180-45] (NL minus FR) to (L_0 cap F);}
  \withColor{\ESF}{\draw[\pgfmathresult,-,out=180+30,in=180-30] (NL cap FR) to (L_0);}
  \withColor{\ESE}{\draw[\pgfmathresult,-] (NL minus FR) to (L_0);}

  \withColor{\ESF}{\draw[\pgfmathresult,-] (L_0 cap F) to (L_0);}
  \withColor{\ESS}{\draw[\pgfmathresult,-] (R_0 cap F) to (R_0);}

  \withColor{\ESS}{\draw[\pgfmathresult,-] (L_0 cap F) to (R_0 cap F);}
  \withColor{\FSS}{\draw[\pgfmathresult,\arrow] (NR cap FL) to (L_0 cap F);}

  \withColor{\ESS}{\draw[\pgfmathresult,-] (NL cap FR) to (R_0 cap F);}

  \withColor{\FSF}{\draw[\pgfmathresult,\arrow] (NR cap FL) to (L_0);}

  \withColor{\ESS}{\draw[\pgfmathresult,-] (NL cap FR) to (R_0);}

  \withColor{\SES}{\draw[\pgfmathresult,\otherarrow] (L_1) to (R_1 cap F);}
  \withColor{\SFS}{\draw[\pgfmathresult,\otherarrow] (NL minus FR) to (R_1 cap F);}
  \withColor{\SFE}{\draw[\pgfmathresult,\otherarrow] (NL minus FR) to (R_1);}
  \withColor{\ESS}{\draw[\pgfmathresult,\otherarrow] (NL minus FR) to (R_0);}
  \withColor{\ESS}{\draw[\pgfmathresult,\otherarrow] (NL minus FR) to (R_0 cap F);}
  \withColor{\ESS}{\draw[\pgfmathresult,\otherarrow] (L_0) to (R_0 cap F);}

  \withColor{\SES}{\draw[\pgfmathresult,\otherarrow] (R_1) to (L_1 cap F);}
  \withColor{\SES}{\draw[\pgfmathresult,\otherarrow] (NR minus FL) to (L_1 cap F);}
  \withColor{\SES}{\draw[\pgfmathresult,\otherarrow] (NR minus FL) to (L_1);}
  \withColor{\FSE}{\draw[\pgfmathresult,\otherarrow] (NR minus FL) to (L_0);}
  \withColor{\FSS}{\draw[\pgfmathresult,\otherarrow] (NR minus FL) to (L_0 cap F);}
  \withColor{\ESS}{\draw[\pgfmathresult,\otherarrow] (R_0) to (L_0 cap F);}
\end{tikzpicture}
\hfill
\begin{tikzpicture}[scale=0.5, every node/.style={scale=0.5}]
  \withColor{\SES}{
    \node[draw, circle, minimum size=1cm,\pgfmathresult] 
      at (0, 0) (L_1 cap F) {$(L' \cap F')_1$};
  }
  \withColor{\SES}{
    \node[draw, circle, minimum size=1cm,\pgfmathresult] 
      at (0, -3) (L_1) {$(L' - F')_1$};
  }
  \withColor{\EFF}{
    \node[draw, ellipse, minimum size=1cm,text width=1.75cm,\pgfmathresult, label=above:{\color{\pgfmathresult} $0$}] 
      at (0, -6) (NL cap FR) {${\inneighborhood{}}_{G'}({L \cap F'},$ $R-F)$};
  }
  \withColor{\EFE}{
    \node[draw, ellipse, minimum size=1cm,text width=1.75cm,\pgfmathresult, label=below:{\color{\pgfmathresult} $0$}] 
      at (0, -9) (NL minus FR) {${\inneighborhood{}}_{G'}({L - F'},$ $R-F)$};
  }
  \withColor{\ESE}{
    \node[draw, circle, minimum size=1cm,\pgfmathresult] 
      at (0, -12) (L_0) {$(L' - F')_0$};
  }
  \withColor{\ESF}{
    \node[draw, circle, minimum size=1cm,\pgfmathresult] 
      at (0, -15) (L_0 cap F) {$(L' \cap F')_0$};
  }

  \withColor{\SEF}{
    \node[draw, circle, minimum size=1cm,\pgfmathresult] 
      at (12, 0) (R_1 cap F) {$(R' \cap F')_1$};
  }
  \withColor{\SEE}{
    \node[draw, circle, minimum size=1cm,\pgfmathresult] 
      at (12, -3) (R_1) {$(R' -F')_1$};
  }
  \withColor{\FEF}{
    \node[draw, ellipse, minimum size=1cm,text width=1.75cm,\pgfmathresult, label=above:{\color{\pgfmathresult} $1$}] 
      at (12, -6) (NR cap FL) {${\inneighborhood{}}_{G'}({R \cap F'},$ $L-F)$};
  }
  \withColor{\FEE}{
    \node[draw, ellipse, minimum size=1cm,text width=1.75cm,\pgfmathresult, label=below:{\color{\pgfmathresult} $1$}] 
      at (12, -9) (NR minus FL) {${\inneighborhood{}}_{G'}({R - F'},$ $L-F)$};
  }
  \withColor{\ESS}{
    \node[draw, circle, minimum size=1cm,\pgfmathresult] 
      at (12, -12) (R_0) {$(R' -F')_0$};
  }
  \withColor{\ESS}{
    \node[draw, circle, minimum size=1cm,\pgfmathresult] 
      at (12, -15) (R_0 cap F) {$(R' \cap F')_0$};
  }
  
  \withColor{\SFE}{
    \node[\pgfmathresult,draw,ellipse,minimum size=1cm,text width=3cm] at
      (6, -3.75) (NCR_1)
      {$\big({\inneighborhood{C}{R-F'}{{G'}}} - \inneighborhood{C}{L-F'}{{G'}}\big)_1$};
  }
  \withColor{\SES}{
    \node[\pgfmathresult,draw,ellipse,minimum size=1cm,text width=3cm] at
      (6, -0) (NCL_1)
      {$\big({\inneighborhood{C}{L - F'}{{G'}}} - \inneighborhood{C}{R-F'}{{G'}}\big)_1$};
  }
  \withColor{\FFE}{
    \node[\pgfmathresult,draw,ellipse,minimum size=1cm,text centered,text width=3cm, label=above:{\color{\pgfmathresult} $1$}] at
      (6, -7.5) (NCLR)
      {${\inneighborhood{C}{L - F'}{{G'}}}$\\$\cap$\\$\inneighborhood{C}{R-F'}{{G'}}$};
  }
  \withColor{\FSE}{
    \node[\pgfmathresult,draw,ellipse,minimum size=1cm,text width=3cm] at
      (6, -11.25) (NCL_0)
      {$\big({\inneighborhood{C}{L - F'}{{G'}}} - \inneighborhood{C}{R-F'}{{G'}}\big)_0$};
  }
  \withColor{\ESS}{
    \node[\pgfmathresult,draw,ellipse,minimum size=1cm,text width=3cm] at
      (6, -15) (NCR_0)
      {$\big({\inneighborhood{C}{R-F'}{{G'}}} - \inneighborhood{C}{L - F'}{{G'}}\big)_0$};
  }
  
  \withColor{\FSE}{\draw[\pgfmathresult,-] (L_0) to (NCL_0);}
  \withColor{\FSE}{\draw[\pgfmathresult,-] (L_0) to (NCLR);}
  \withColor{\FSF}{\draw[\pgfmathresult,-] (L_0 cap F) to (NCL_0);}
  \withColor{\FSF}{\draw[\pgfmathresult,-] (L_0 cap F) to (NCLR);}
  \withColor{\ESS}{\draw[\pgfmathresult,-] (L_0 cap F) to (NCR_0);}
  \withColor{\SSF}{\draw[\pgfmathresult,\arrow] (L_0 cap F) to (NCR_1);}
  
  \withColor{\SES}{\draw[\pgfmathresult,-] (L_1) to (NCL_1);}
  \withColor{\SFS}{\draw[\pgfmathresult,\arrow] (NCLR) to (L_1);}
  \withColor{\SES}{\draw[\pgfmathresult,-] (L_1 cap F) to (NCL_1);}
  \withColor{\SFS}{\draw[\pgfmathresult,\arrow] (NCR_1) to (L_1 cap F);}
  \withColor{\SFS}{\draw[\pgfmathresult,\arrow] (NCLR) to (L_1 cap F);}
  
  \withColor{\FSF}{\draw[\pgfmathresult,-] (NL cap FR) to (NCL_0);}
  \withColor{\SFS}{\draw[\pgfmathresult,\arrow] (NL cap FR) to (NCL_1);}
  \withColor{\ESS}{\draw[\pgfmathresult,-] (NL cap FR) to (NCR_0);}
  \withColor{\SSF}{\draw[\pgfmathresult,\arrow] (NL cap FR) to (NCR_1);}
  \withColor{\FSF}{\draw[\pgfmathresult,-] (NL cap FR) to (NCLR);}
  
  \withColor{\FSE}{\draw[\pgfmathresult,-] (NL minus FR) to (NCL_0);}
  \withColor{\SFS}{\draw[\pgfmathresult,\arrow] (NL minus FR) to (NCL_1);}
  \withColor{\FSE}{\draw[\pgfmathresult,-] (NL minus FR) to (NCLR);}

  \withColor{\ESS}{\draw[\pgfmathresult,-] (R_0) to (NCR_0);}
  \withColor{\FSS}{\draw[\pgfmathresult,\arrow] (NCLR) to (R_0);}
  \withColor{\FSS}{\draw[\pgfmathresult,\arrow] (NCLR) to (R_0 cap F);}
  \withColor{\FSS}{\draw[\pgfmathresult,\arrow] (NCL_0) to (R_0 cap F);}
  \withColor{\ESS}{\draw[\pgfmathresult,-] (R_0 cap F) to (NCR_0);}

  \withColor{\SFE}{\draw[\pgfmathresult,-] (R_1) to (NCR_1);}
  \withColor{\SFE}{\draw[\pgfmathresult,-] (R_1) to (NCLR);}
  \withColor{\SSF}{\draw[\pgfmathresult,\arrow] (R_1 cap F) to (NCL_0);}
  \withColor{\SFF}{\draw[\pgfmathresult,-] (R_1 cap F) to (NCR_1);}
  \withColor{\SFF}{\draw[\pgfmathresult,-] (R_1 cap F) to (NCLR);}
  \withColor{\SES}{\draw[\pgfmathresult,-] (NCL_1) to (R_1 cap F);}

  \withColor{\SFF}{\draw[\pgfmathresult,-] (NR cap FL) to (NCR_1);}
  \withColor{\FSS}{\draw[\pgfmathresult,\arrow] (NR cap FL) to (NCR_0);}
  \withColor{\SES}{\draw[\pgfmathresult,-] (NR cap FL) to (NCL_1);}
  \withColor{\SSF}{\draw[\pgfmathresult,\arrow] (NR cap FL) to (NCL_0);}
  \withColor{\SFF}{\draw[\pgfmathresult,-] (NR cap FL) to (NCLR);}

  \withColor{\SFE}{\draw[\pgfmathresult,-] (NR minus FL) to (NCR_1);}
  \withColor{\FSS}{\draw[\pgfmathresult,\arrow] (NR minus FL) to (NCR_0);}
  \withColor{\SFE}{\draw[\pgfmathresult,-] (NR minus FL) to (NCLR);}

  \withColor{\SES}{\draw[\pgfmathresult,\otherarrow] (R_1) to (NCL_1);}
  \withColor{\SSE}{\draw[\pgfmathresult,\otherarrow] (R_1) to (NCL_0);}
  \withColor{\SES}{\draw[\pgfmathresult,\otherarrow] (NR minus FL) to (NCL_1);}
  \withColor{\SSE}{\draw[\pgfmathresult,\otherarrow] (NR minus FL) to (NCL_0);}
  \withColor{\ESS}{\draw[\pgfmathresult,\otherarrow] (NL minus FR) to (NCR_0);}
  \withColor{\SSE}{\draw[\pgfmathresult,\otherarrow] (NL minus FR) to (NCR_1);}
  \withColor{\ESS}{\draw[\pgfmathresult,\otherarrow] (L_0) to (NCR_0);}
  \withColor{\SSE}{\draw[\pgfmathresult,\otherarrow] (L_0) to (NCR_1);}
\end{tikzpicture}
\\
~
\\
\begin{tikzpicture}[scale=0.55, every node/.style={scale=0.55}]
  \withColor{\SES}{
    \node[draw, circle, minimum size=1cm,\pgfmathresult]
    at (0, 0) (L_1 cap F) {$(L' \cap F')_1$};
  }
  \withColor{\SES}{
    \node[draw, circle, minimum size=1cm,\pgfmathresult]
    at (0, -3) (L_1) {$(L' - F')_1$};
  }
  \withColor{\EFF}{
    \node[draw, ellipse, minimum size=1cm,text width=1.75cm,\pgfmathresult, label=above:{\color{\pgfmathresult} $0$}]
    at (0, -6) (NL cap FR) {${\inneighborhood{}}_{G'}({L \cap F'},$ $R-F)$};
  }
  \withColor{\EFE}{
    \node[draw, ellipse, minimum size=1cm,text width=1.75cm,\pgfmathresult, label=below:{\color{\pgfmathresult} $0$}]
      at (0, -9) (NL minus FR) {${\inneighborhood{}}_{G'}({L - F'},$ $R-F)$};
  }
  \withColor{\ESE}{
    \node[draw, circle, minimum size=1cm,\pgfmathresult]
    at (0, -12) (L_0) {$(L' - F')_0$};
  }
  \withColor{\ESF}{
    \node[draw, circle, minimum size=1cm,\pgfmathresult]
    at (0, -15) (L_0 cap F) {$(L' \cap F')_0$};
  }
  \withColor{\SEF}{
    \node[draw, circle, minimum size=1cm,\pgfmathresult]
    at (8, 0) (R_1 cap F) {$(R' \cap F')_1$};
  }
  \withColor{\SEE}{
    \node[draw, circle, minimum size=1cm,\pgfmathresult]
    at (8, -3) (R_1) {$(R' -F')_1$};
  }
  \withColor{\FEF}{
    \node[draw, ellipse, minimum size=1cm,text width=1.75cm,\pgfmathresult, label=above:{\color{\pgfmathresult} $1$}]
    at (8, -6) (NR cap FL) {${\inneighborhood{}}_{G'}({R \cap F'},$ $L-F)$};
  }
  \withColor{\FEE}{
    \node[draw, ellipse, minimum size=1cm,text width=1.75cm,\pgfmathresult, label=below:{\color{\pgfmathresult} $1$}]
    at (8, -9) (NR minus FL) {${\inneighborhood{}}_{G'}({R - F'},$ $L-F)$};
  }
  \withColor{\ESS}{
    \node[draw, circle, minimum size=1cm,\pgfmathresult]
    at (8, -12) (R_0) {$(R' -F')_0$};
  }
  \withColor{\ESS}{
    \node[draw, circle, minimum size=1cm,\pgfmathresult]
    at (8, -15) (R_0 cap F) {$(R' \cap F')_0$};
  }
  
  \withColor{\ESS}{
    \node[\pgfmathresult,draw, circle, minimum size=1cm] at
      (4, -13) (C_0) {$C'_0$};
  }
  \withColor{\SSE}{
    \node[\pgfmathresult,draw, circle, minimum size=1cm, label=below right:{\color{\pgfmathresult} $1$}] at
      (4, -7.5) (C'_1) {$C'_2$};
  }
  \withColor{\SES}{
    \node[\pgfmathresult,draw, circle, minimum size=1cm] at
      (4, -2) (C_1) {$C'_1$};
  }
  
  \withColor{\SES}{\draw[\pgfmathresult,-] (L_1 cap F) to (C_1);}
  \withColor{\ESS}{\draw[\pgfmathresult,-] (L_0 cap F) to (C_0);}
  \withColor{\SSF}{\draw[\pgfmathresult,\arrow] (L_0 cap F) to (C'_1);}
  \withColor{\ESS}{\draw[\pgfmathresult,-] (NL cap FR) to (C_0);}
  \withColor{\SFS}{\draw[\pgfmathresult,\arrow] (NL cap FR) to (C_1);}
  \withColor{\SSF}{\draw[\pgfmathresult,\arrow] (NL cap FR) to (C'_1);}

  \withColor{\SES}{\draw[\pgfmathresult,-] (R_1 cap F) to (C_1);}
  \withColor{\ESS}{\draw[\pgfmathresult,-] (R_0 cap F) to (C_0);}
  \withColor{\SSF}{\draw[\pgfmathresult,\arrow] (R_1 cap F) to (C'_1);}
  \withColor{\SES}{\draw[\pgfmathresult,-] (NR cap FL) to (C_1);}
  \withColor{\FSS}{\draw[\pgfmathresult,\arrow] (NR cap FL) to (C_0);}
  \withColor{\SSF}{\draw[\pgfmathresult,\arrow] (NR cap FL) to (C'_1);}

  \withColor{\SES}{\draw[\pgfmathresult,\otherarrow] (L_1) to (C_1);}
  \withColor{\ESS}{\draw[\pgfmathresult,\otherarrow] (L_0) to (C_0);}
  \withColor{\SSE}{\draw[\pgfmathresult,\otherarrow] (L_0) to (C'_1);}
  \withColor{\SFS}{\draw[\pgfmathresult,\otherarrow] (NL minus FR) to (C_1);}
  \withColor{\ESS}{\draw[\pgfmathresult,\otherarrow] (NL minus FR) to (C_0);}
  \withColor{\SSE}{\draw[\pgfmathresult,\otherarrow] (NL minus FR) to (C'_1);}

  \withColor{\ESS}{\draw[\pgfmathresult,\otherarrow] (R_0) to (C_0);}
  \withColor{\SES}{\draw[\pgfmathresult,\otherarrow] (R_1) to (C_1);}
  \withColor{\SSE}{\draw[\pgfmathresult,\otherarrow] (R_1) to (C'_1);}
  \withColor{\SES}{\draw[\pgfmathresult,\otherarrow] (NR minus FL) to (C_1);}
  \withColor{\FSS}{\draw[\pgfmathresult,\otherarrow] (NR minus FL) to (C_0);}
  \withColor{\SSE}{\draw[\pgfmathresult,\otherarrow] (NR minus FL) to (C'_1);}

\draw[white,-,out=180+35,in=180-35] (NL cap FR) to (L_0 cap F);
\end{tikzpicture}
\hfill
\begin{tikzpicture}[scale=0.45, every node/.style={scale=0.45}, rotate=90]
  \withColor{\ESS}{
    \node[\pgfmathresult, draw, circle, minimum size=1cm] at
(6,-4) (C_0) {$C'_0$};
  }
  \withColor{\SSE}{
    \node[\pgfmathresult, draw, circle, minimum size=1cm, label=above:{\color{\pgfmathresult} $1$}] at 
      (180+180:6) (C'_1) {$C'_2$};
  }
  \withColor{\SES}{
    \node[\pgfmathresult, draw, circle, minimum size=1cm] at 
(6,4) (C_1) {$C'_1$};
  }
  
  \withColor{\FSE}{
    \node[draw, ellipse, minimum width=3.75cm,text width=3cm,\pgfmathresult]
      at (270+180:6) (NCL_0) {
        $\big({\inneighborhood{C}{L - F'}{{G'}}} -$ \\ ${\inneighborhood{C}{R - F'}{{G'}}}\big)_0$
      };
  }
\withColor{\ESS}{
    \node[draw, ellipse, minimum width=3.75cm,text width=3cm,\pgfmathresult]
at (-6,6) (NCR_0) {
        $\big({\inneighborhood{C}{R - F'}{{G'}}} -$ \\ ${\inneighborhood{C}{L - F'}{{G'}}}\big)_0$
      };
  }
  \withColor{\FFE}{
\node[draw, ellipse, minimum size=1cm,text centered,text width=3cm,\pgfmathresult, label=below:{\color{\pgfmathresult} $1$}]
      at (360+180:6) (NCLR) {
        ${\inneighborhood{C}{L - F'}{{G'}}}$ \\ $\cap$ \\ ${\inneighborhood{C}{R - F'}{{G'}}}$
      };
  }
  \withColor{\SES}{
\node[draw, ellipse, minimum width=3.75cm,text width=3cm,\pgfmathresult]
at (-6,-6) (NCL_1) {
        $\big({\inneighborhood{C}{L - F'}{{G'}}} -$ \\ ${\inneighborhood{C}{R - F'}{{G'}}}\big)_1$
      };
  }
  \withColor{\SFE}{
\node[draw, ellipse, minimum width=3.75cm,text width=3cm,\pgfmathresult]
      at (90+180:6) (NCR_1) {
        $\big({\inneighborhood{C}{R - F'}{{G'}}} -$ \\ ${\inneighborhood{C}{L - F'}{{G'}}}\big)_1$
      };
  }
  
  \withColor{\FSS}{\draw[\pgfmathresult, \arrow] (NCLR) to (C_0);}
  \withColor{\SSE}{\draw[\pgfmathresult, -] (NCLR) to (C'_1);}
  \withColor{\SFS}{\draw[\pgfmathresult, \arrow] (NCLR) to (C_1);}
  \withColor{\SSE}{\draw[\pgfmathresult, -] (NCLR) to (NCL_0);}
  \withColor{\SFS}{\draw[\pgfmathresult, \arrow] (NCLR) to (NCL_1);}
  \withColor{\FSS}{\draw[\pgfmathresult, \arrow] (NCLR) to (NCR_0);}
  \withColor{\SSE}{\draw[\pgfmathresult, -] (NCLR) to (NCR_1);}
  
  \withColor{\FSS}{\draw[\pgfmathresult, \arrow] (NCL_0) to (C_0);}
  \withColor{\SSE}{\draw[\pgfmathresult, -] (NCL_0) to (C'_1);}
  \withColor{\FSS}{\draw[\pgfmathresult, \arrow] (NCL_0) to (NCR_0);}
  \withColor{\SSE}{\draw[\pgfmathresult, -] (NCL_0) to (NCR_1);}
  
  \withColor{\SFS}{\draw[\pgfmathresult, \arrow] (NCR_1) to (C_1);}
  \withColor{\SSE}{\draw[\pgfmathresult, -] (NCR_1) to (C'_1);}
  \withColor{\SFS}{\draw[\pgfmathresult, \arrow] (NCR_1) to (NCL_1);}
  
  \withColor{\ESS}{\draw[\pgfmathresult, -] (NCR_0) to (C_0);}
  \withColor{\SES}{\draw[\pgfmathresult, -] (NCL_1) to (C_1);}
\end{tikzpicture}   \caption{
    ${\underlying{\mathcal{G}}}$
    to model $E_1$, $E_2$, and $E_3$.
    The numbers adjacent to the sets are the corresponding inputs in
    execution $\mathcal E$;
    if there is no number adjacent to the set,
    then the input is the same as the subscript.
An undirected edge denotes that edges can exist in both directions.
    A directed edge with a hollow arrow denotes that edges could only exist in
    one direction between the original nodes in $\underlying{G'}$.
    A directed edge with a solid arrow denotes that edges could have existed in
    both directions between the original nodes in $\underlying{G'}$,
    but in ${\underlying{\mathcal{G}}}$,
    they only exist in one direction between the two copies.
  }
  \label{figure hypergraph necessity network}
\end{figure}
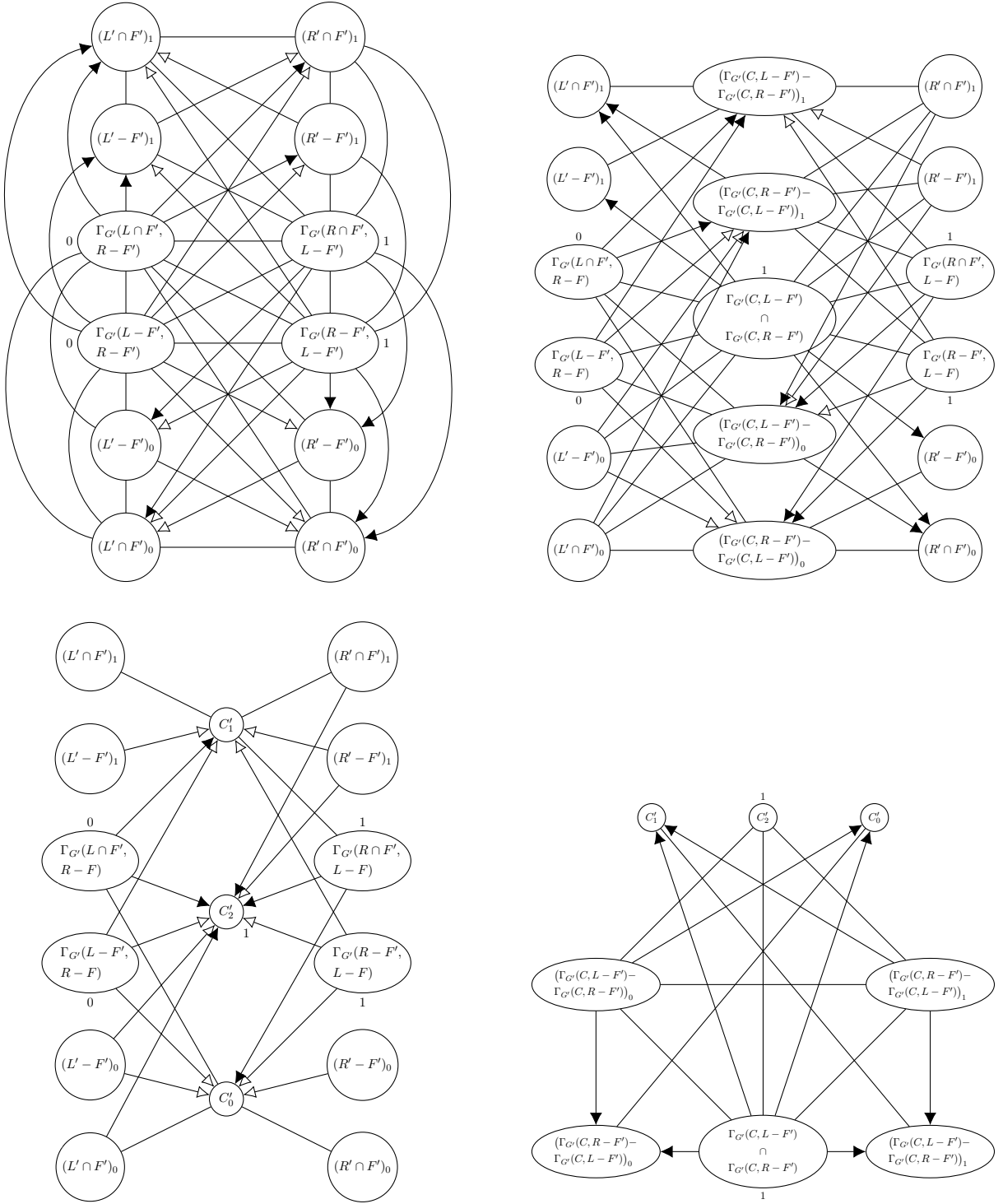

\begin{figure}[p]
  \centering
  \setcounter{index}{1}
 

  \caption{
    Execution $E_{\theindex}$
    as modeled by execution $\Sigma$ on hypergraph $\mathcal{G}$
    (the underlying simple graph $\underlying{\mathcal{G}}$ is shown here).
    The red nodes are faulty in $E_{\theindex}$.
    The gray network is simulated by the faulty nodes
    (this includes edges between faulty nodes).
    Edges between faulty and non-faulty nodes are
    depicted in red.
  }
  \label{figure hypergraph necessity exec 0}
\end{figure}

\begin{figure}[p]
  \centering
  \setcounter{index}{2}
  %
 

  \caption{
    Execution $E_{\theindex}$
    as modeled by execution $\Sigma$ on hypergraph $\mathcal{G}$
    (the underlying simple graph $\underlying{\mathcal{G}}$ is shown here).
    The red nodes are faulty in $E_{\theindex}$.
    The gray network is simulated by the faulty nodes.
    Edges between faulty and non-faulty nodes are
    depicted in red.
  }
  \label{figure hypergraph necessity exec 1}
\end{figure}
  
\begin{figure}[p]
  \centering
  \setcounter{index}{3}
  %
 

  \caption{
    Execution $E_{\theindex}$
    as modeled by execution $\Sigma$ on hypergraph $\mathcal{G}$
    (the underlying simple graph $\underlying{\mathcal{G}}$ is shown here).
    The red nodes are faulty in $E_{\theindex}$.
    The gray network is simulated by the faulty nodes.
    Edges between faulty and non-faulty nodes are
    depicted in red.
  }
  \label{figure hypergraph necessity exec 2}
\end{figure}

  In execution $\Sigma$,
  nodes in
  $(L' - F')_0$ and $\inneighborhood{L-F'}{R-F'}{{G'}}$ output $0$
  while
  nodes in
  $(R' - F')_1$ and $\inneighborhood{R-F'}{L-F'}{{G'}}$ output $1$.
  Observe that these copies model the nodes in $L - F'$ and $R - F'$,
  respectively, in ${G'}$.
  Therefore, in execution $E_3$,
  nodes in $L - F'$ output $0$ while
  nodes in $R - F'$ output $1$.
  Recall that both these sets are non-empty by construction.
  Thus algorithm $\mathcal A'$ in execution $E_3$ on hypergraph $G'$
  terminates without agreement between these two sets of nodes,
  a contradiction.
\end{proof}

\section[Proof of Correctness of Algorithm \ref*{algorithm hypergraphs}]
{Proof of Correctness of \algoRef{algorithm hypergraphs}}
\label{section proof sufficiency multiple}

In this section, we show correctness of
\algoRef{algorithm hypergraphs}
when the hypergraph $G$ satisfies \conditionSC{}.
For the rest of this section,
we assume that $G$
satisfies both \conditionSC{} and \conditionNC{}
(recall that, by \theoremRef{theorem hypergraphs menger},
the two conditions are equivalent).
Throughout this section,
we use $F^*$ to denote the actual set of faulty nodes.
We prove \lemmaRef{lemma hypergraphs validity} first.

\begin{proof}[Proof of \lemmaRef{lemma hypergraphs validity}]
  Fix a phase $> 0$.
  We use $\gamma_u^\text{start}$ and $\gamma_u^\text{end}$
  to denote the state $\gamma_u$ of node $u$ at the
  beginning and end of the phase, respectively.
  Consider an arbitrary non-faulty node $v$.
  If $v$ does not update its state in this phase,
  then the claim is trivially true since
  $\gamma_v^\text{end} = \gamma_v^\text{start}$.
  So suppose that $v$ did update its state in this phase.
  Then it must have done so in either \step{d} or \step{e} (but not both).
  We consider each case separately.

  \begin{enumerate}[label={Case \arabic*:},leftmargin=*]
    \item $v \in S$ updated its state in \step{d}.\\
    Suppose $v$ updated its state $\gamma_v$ to $\tau \in \set{0, 1}$ in \step{d}.
    Then, as per the update rules in \step{d},
    $v$ must have received the value $\tau$ identically
    along $f+1$ node-disjoint $A_v v$-paths in \step{b}.
Since there are at most $f$ faulty nodes,
    at least one of the $A_v v$-paths, say $P$,
    must neither have any faulty internal node nor a faulty source node.
    Now $\tau$ was received along $P$,
    which has exclusively non-faulty internal nodes.
    So the source node of $P$, say $u$,
    flooded $\tau$ in \step{b} of this phase.
Furthermore, $u$ is non-faulty.
    Thus, $\gamma_u^\text{start} = \tau$
    at the start of this phase.
    Therefore, the state of node $v$ at the end of this phase
    equals the state of a non-faulty node $u$
    at the start of this phase.

    \item $v \in V - S - F$ updated its state in \step{f}.\\
    Suppose $v$ updated its state $\gamma_v$ to $\tau \in \set{0, 1}$ in \step{f}.
    Then, as per the update rules in \step{f},
    $v$ must have received the value $\tau$ identically
    along $f+1$ node-disjoint $S v$-paths in \step{e}.
Since there are at most $f$ faulty nodes,
    at least one of the $S v$-paths, say $P$,
    must neither have any faulty internal node nor a faulty source node.
    Now $\tau$ was received along $P$,
    which has exclusively non-faulty internal nodes.
    So the source node of $P$, say $u$,
    flooded $\tau$ in \step{e} of this phase.
Note that $u \in S$.
    If $u$ did not update its state in \step{d},
    then $\tau = \gamma_u^\text{end} = \gamma_u^\text{start}$.
    Otherwise, by Case 1 above,
    $\tau = \gamma_u^\text{end} = \gamma_w^\text{start}$
    for some non-faulty node $w$.
    In both cases, $\tau$ is a $\gamma$ value of some non-faulty node
    at the start of this phase.
    Therefore, the state of node $v$ at the end of this phase
    equals the state of some non-faulty node
    at the start of this phase.
  \end{enumerate}
  In both cases, we have that $\gamma_v^\text{end} = \gamma_u^\text{start}$
  for some non-faulty node $u$.
\end{proof}

Before proving \lemmaRef{lemma hypergraphs agreement},
we need some intermediate results.
First, we show the proofs of Lemmas
\ref{lemma hypergraphs source unique},
\ref{lemma hypergraphs source SC}, and
\ref{lemma hypergraphs source propagates},
which are similar to Lemmas 6, 7, and 10 in \cite{khan2020directedOPODIS}.

\begin{proof}[Proof of \lemmaRef{lemma hypergraphs source unique}]
  Fix an arbitrary set $F$.
	Suppose, for the sake of contradiction, that
  the directed decomposition of $G - F$
  has two source components $S_1$ and $S_2$.
  To derive the contradiction,
  we show that $G$ does not satisfy \conditionNC{}.
  Let
  \begin{align*}
    L &:= S_1, \\
    R &:= S_2 \cup F,  \\
    C &:= V - S_1 - S_2 - F,
  \end{align*}
  so that $(L, R, C)$ is a partition of $V$.
	Observe that both $L - F = S_1$ and $R-F = S_2$ are non-empty.
Since $S_1$ is a source component
  of the directed decomposition of $G -F$,
  \begin{align*}
    \inneighborhood{R \cup C}{L - F}{G}
      &=  \inneighborhood{V - S_1}{S_1}{G}
      &\text{since $R \cup C = V - S_1$ and $L - F = S_1$} \\
      &\subseteq  F
      &\text{since $S_1$ is a source component of $G - F$.}
  \end{align*}
  Similarly,
  since $S_2$ is also a source component
  of the directed decomposition of $G -F$,
  \begin{align*}
    \inneighborhood{L \cup C}{R - F}{G}
      &=  \inneighborhood{V - S_2 - F}{S_2}{G}
      &\text{since $L \cup C = V - S_2 - F$ and $R - F = S_2$}  \\
      &=  \emptyset
      &\text{since $S_2$ is a source component of $G - F$.}
  \end{align*}
  Therefore,
  \begin{enumerate}[label=\arabic*)]
    \item
$\abs{\inneighborhood{L \cup C}{R - F}{G}}
      = 0 \le f
\implies L \cup C
        \notadjacent{G} R-F$,
and
    \item
$\abs{\inneighborhood{R \cup C}{L - F}{G}}
      \le \abs{F} \le f
\implies R \cup C
        \notadjacent{G} L-F$.
\end{enumerate}
  Note that $F' = F$ for $G' = G$.
  Since $G \in \graphSplitFSet_F(G)$,
  this violates \conditionNC{}, a contradiction.
\end{proof}

\begin{proof}[Proof of \lemmaRef{lemma hypergraphs source SC}]
  Fix an arbitrary set $F$.
  Let $S$ be the unique source component
  in the directed decomposition of $G - F$,
  and let $\Phi = \inneighborhood{F}{S}{G}$.
  Suppose, for the sake of contradiction, that
  $G[S \cup \Phi]$ does not satisfy \conditionSC{}
  with parameter $\Phi$.
  Then, by \theoremRef{theorem hypergraphs menger},
  $G[S \cup \Phi]$ does not satisfy \conditionNC{}
  with parameter $\Phi$ either.
  So there exists a hypergraph
  $G_\Phi \in \graphSplitFSet_\Phi(G[S \cup \Phi])$
  and,
  using $\Phi'$ to denote the set of nodes in $G_\Phi$
  corresponding to nodes in $\Phi$ in $G$,
  a partition $(L, C, R)$ of $S \cup \Phi'$
  such that
  \begin{enumerate}[label=\arabic*)]
    \item $L \cup C \notadjacent{G_\Phi} R - \Phi'$, and
    \item $R \cup C \notadjacent{G_\Phi} L - \Phi'$.
  \end{enumerate}
  Observe that this implies that both $R - \Phi'$ and $L - \Phi'$ are non-empty,
  by definition of $\adjacent{}$.

  Since $\Phi \subseteq F$,
  so there exists a hypergraph $G' \in \graphSplitFSet_F(G)$
  that is obtained by splitting \emph{exactly} the same
  nodes as were split to obtain $G_\Phi$,
  and making the same assignments in the split operations
  in both graphs.
  So $G'$ has the node set $V' = \Phi' \cup (V - \Phi)$,
  and $G_\Phi = G'[S \cup \Phi']$.
  Let $F'$ denote the set of nodes in $G'$
  corresponding to nodes in $F$ in $G$
  (i.e., $F' = \Phi' \cup (F - \Phi)$),
  and let
  \begin{align*}
    C'  &:= C \cup (V - S - \Phi).
  \end{align*}
  Then $(L, C', R)$ is a partition of $V'$.
  To complete the contradiction,
  we show that $L \cup C' \notadjacent{G'} R - F'$
  and $R \cup C' \notadjacent{G'} L - F'$,
  which violates \conditionNC{}.

  We first show that in $G'$,
  nodes in $(L \cup R) - F'$
  have no in-neighbors in
  $C' - C$,
  as follows.
  \begin{align*}
    \inneighborhood{C' - C}{(L \cup R) - F'}{G'}
      &=  \inneighborhood{V - S - \Phi}{(L \cup R) - F'}{G'}
      &\text{since $C' - C = V - S - \Phi$} \\
      &\subseteq \inneighborhood{V - S - \Phi}{S}{G'}
      &\text{since $L \cup R \subseteq S \cup \Phi \subseteq S \cup F'$}\\
      &=\inneighborhood{V - S - \Phi}{S}{G}
      &\text{since $S, V - S - \Phi \subseteq V$} \\
      &=\emptyset,
  \end{align*}
  where the last equality follows from the fact that
  $S$ is the source component
  in the directed decomposition of $G - F$.
  Now, we have
  \begin{align*}
    \abs{\inneighborhood{L \cup C'}{R - F'}{G'}}
      &=  \abs{\inneighborhood{L \cup C}{R - F'}{G'}}
      &\text{since $\inneighborhood{C' - C}{R - F'}{G'} = \emptyset$}
      \\
      &=  \abs{\inneighborhood{L \cup C}{R - \Phi'}{G'}}
      &\text{since $R \cap F' \subseteq (S \cup \Phi') \cap F' = \Phi'$}
      \\
      &=  \abs{\inneighborhood{L \cup C}{R - \Phi'}{G_\Phi}}
      &\text{since $G'[L \cup R \cup C] = G'[S \cup \Phi'] = G_\Phi$}
      \\
      &\le f
      &\text{since $L \cup C \notadjacent{G_\Phi} R - \Phi'$}
      \\
    \implies L \cup C' &\notadjacent{G'} R - F'
      &\text{since $R - F' = R - \Phi' \ne \emptyset$.}
  \end{align*}
  Similarly,
  \begin{align*}
    \abs{\inneighborhood{R \cup C'}{L - F'}{G'}}
      &=  \abs{\inneighborhood{R \cup C}{L - F'}{G'}}
      &\text{since $\inneighborhood{C' - C}{L - F'}{G'} = \emptyset$}
      \\
      &=  \abs{\inneighborhood{R \cup C}{L - \Phi'}{G'}}
      &\text{since $L \cap F' \subseteq (S \cup \Phi') \cap F' = \Phi'$}
      \\
      &=  \abs{\inneighborhood{R \cup C}{L - \Phi'}{G_\Phi}}
      &\text{since $G'[L \cup R \cup C] = G'[S \cup \Phi'] = G_\Phi$}
      \\
      &\le f
      &\text{since $R \cup C \notadjacent{G_\Phi} L - \Phi'$}
      \\
    \implies R \cup C' &\notadjacent{G'} L - F'
      &\text{since $L - F' = L - \Phi' \ne \emptyset$.}
  \end{align*}
  This violates \conditionNC{}, a contradiction.
\end{proof}

\begin{proof}[Proof of \lemmaRef{lemma hypergraphs source propagates}]
  Fix an arbitrary set $F$.
  Let $S$ be the unique source component
  in the directed decomposition of $G - F$.
  Let
  \begin{align*}
    A &:= S,  \\
    B &:= V - S,
  \end{align*}
  so that $(A, B)$ is a partition of $V$.
	Now, since $A = S$ is the source component in the
  directed decomposition of $G - F$,
	we have
  \begin{align*}
    \inneighborhood{B}{A}{G}
      &= \inneighborhood{V - S}{S}{G}  \\
      &= \inneighborhood{F}{S}{G}  \\
      &\subseteq F.
  \end{align*}
  That is,
  all the in-neighbors of $A$ in $B$ are contained entirely in $F$.
So, by Menger's Theorem,
  for any node $v \in A$,
  there can be at most $\abs{F} \le f$ node-disjoint $Bv$-paths in $G$.
	Thus, $B \notpropagate{G - (A \cap F)} A - F$.
	Since $G \in \graphSplitFSet_F(G)$,
  by \conditionSC{},
  we have \[S = A \propagate{G - (B \cap F)} B - F = V - S - F.\]
	The result follows from the fact that $B \cap F = F$.
\end{proof}

Now, we show that in every iteration of the main \texttt{for} loop
of \algoRef{algorithm hypergraphs},
the paths in \step{c}
do exist.

\begin{lemma} \label{lemma algorithm nodes connected multiple}
  In any phase $>0$ of \algoRef{algorithm hypergraphs}
  with a candidate faulty set $F$,
  for any two nodes $u, v \in S$,
  there exists a $uv$-path in $G[S]$.
\end{lemma}
\begin{proof}
  Immediately since $S$ is strongly connected.
\end{proof}

In the flooding procedure
(\cite{khan2019undirectedPODC, khan2020directedOPODIS}),
when a non-faulty node wants to flood a value $b \in \set{0,1}$,
it sends a single value $b$ on all of its
hyperedges.
But a faulty node might send different messages
on different hyperedges.
Note, however, that even a faulty node must send the exact same value
on a single hyperedge:
if it sends two different values on the same hyperedge,
then the receiving nodes can choose the first value and ignore the later one.

\begin{lemma} \label{lemma algorithm correct split graph multiple}
  Consider a phase $> 0$ of \algoRef{algorithm hypergraphs}
  wherein $F = F^*$.
  For any two non-faulty nodes $u, v \in S$,
  we have $G'_u = G'_v$
  in \step{c} of this phase.
  Furthermore,
  if in \step{b} of this phase
  a faulty node $z \in \inneighborhood{F^*}{S}{G}$ transmitted $0$ (resp. $1$)
  on a hyperedge $e \in \incidentedges_G(z)$,
  such that $\tail(e) \cap S$ is non-empty,
  then
  in \step{c} of this phase
  $e$ is assigned to $z^0$ (resp. $z^1$)
  in $G'_u = G'_v$.
\end{lemma}
\begin{proof}
  Consider the phase where $F = F^*$
  and any two non-faulty nodes $u, v \in S$.
  Observe that the node set of the two hypergraphs
  $G'_u$ and $G'_v$ are the same.
  For the hyperedges, by construction,
  it is sufficient to show that,
  for any $z \in F^*$,
  the assignment of multicast channels to $z^0$ and $z^1$
  in the split operation
  is the same in $G'_u$ as in $G'_v$.
  Consider an arbitrary node $z \in F^*$
  and a hyperedge $e \in \incidentedges_{G}(z)$.
  There are two cases to consider:

  \begin{enumerate}[label={Case \arabic*:},leftmargin=*]
    \item There exists a node $w \in \tail(e)$ such that $w \in S$.\\
    So $z \in \inneighborhood{F^*}{S}{G}$.
    By \lemmaRef{lemma algorithm nodes connected multiple},
    there exists a $wu$-path and a $wv$-path in $G[S]$.
    Therefore, there exists a $zu$-path and a $zv$-path in
    $G[S, \set{e}]$.
    Let $P_{zu}$ and $P_{zv}$ be the $zu$-path and $zv$-path,
    respectively, identified by nodes $u$ and $v$ in \step{c}.
    Observe that, for both these paths,
    the first hyperedge on the path is $e$ and
    $z$ is the only faulty node.
    Since $z$ is the source node in $P_{zu}$ and $P_{zv}$,
    both these paths do not have any faulty internal node.
    Therefore,
    in \step{b},
    if $z$ transmitted $0$ on hyperedge $e$,
    then $u$ (resp. $v$) received value $0$
    along $P_{zu}$ (resp. $P_{zv}$).
    So both $u$ and $v$ assign $e$ to $z^0$
    in $G'_u$ and $G'_v$, respectively.
    Similarly, if $z$ transmitted $1$ on hyperedge $e$ in \step{b},
    then both $u$ and $v$ assign $e$ to $z^1$
    in $G'_u$ and $G'_v$, respectively.

    \item There does not exist any node $w \in \tail(e)$
    such that $w \in S$.\\
    Then
    there is no $zu$-path or $zv$-path in
    $G[S, \set{e}]$.
    Therefore, both $u$ and $v$ assign $e$ to $z^1$
    in $G'_u$ and $G'_v$, respectively.
  \end{enumerate}
  In both cases,
  we have that the hyperedge $e$ was assigned identically
  by both $u$ and $v$.
  Observe that if $z \in \inneighborhood{F^*}{S}{G}$ and
  $z$ transmitted $0$ (resp. $1$) on a hyperedge $e \in \incidentedges_G(z)$,
  such that $\tail(e) \cap S$ is non-empty,
  then $e$ is assigned to $z^0$ (resp. $z^1$)
  by both $u$ and $v$, as required.
\end{proof}

\begin{lemma} \label{lemma hypergraphs alg correct partition}
  Consider a phase $> 0$ of \algoRef{algorithm hypergraphs}
  wherein $F = F^*$.
  Let
  \begin{align*}
    Z &:= \set{u^0 \mid u \in \inneighborhood{F^*}{S}{G}} \cup
    \set{ w \in S \mid
    \text{$w$ flooded value $0$ in \step{b} of this phase}},\\
    N &:= \set{u^1 \mid u \in \inneighborhood{F^*}{S}{G}} \cup
    \set{ w \in S \mid
    \text{$w$ flooded value $1$ in \step{b} of this phase}}.
  \end{align*}
  For any two non-faulty nodes $u, v \in S$,
  we have $Z = Z_u = Z_v$ and $N = N_u = N_v$
  in \step{c} of this phase.
\end{lemma}
\begin{proof}
  Consider the phase where $F = F^*$
  and $S \subseteq V - F^*$ is the unique source component
  in the directed decomposition of $G - F^*$.
  For any two non-faulty nodes $u, v \in S$,
  we show that $Z \subseteq Z_v$ and $N \subseteq N_v$
  (resp. $Z \subseteq Z_u$ and $N \subseteq N_u$).
  Since $Z \cup N = Z_u \cup N_u = Z_v \cup N_v$,
  it follows that $Z = Z_u = Z_v$ and $N = N_u = N_v$.
  For a node $w \in \inneighborhood{F^*}{S}{G}$,
  the two split nodes $w^0$ and $w^1$ are assigned identically by both
  $u$ and $v$.
  So consider an arbitrary node $w \in S$.
  Recall that we are considering the phase $> 0$ of the algorithm
  where $F = F^*$ is the actual set of faulty nodes.
  So $w$ is non-faulty.
  There are two cases to consider:

  \begin{enumerate}[label={Case \arabic*:},leftmargin=*]
    \item
    $w \in Z - F^*$, i.e., $w \in S$ flooded $0$ in \step{b}
    of this phase.\\
    Let $P_{wv}$ be the $wv$-path identified by $v$ in \step{c}.
    Note that $P_{wv}$ is contained entirely in $G - F^*$
    so that
    $P_{wv}$ does not have any faulty nodes.
    It follows that, in \step{b},
    $w$ flooded the value $0$.
    So $v$ received value $0$ along $P_{wv}$.
    Therefore,
    in \step{c}
    $v$ puts $w$ in the set $Z_v$.

    \item
    $w \in N - F^*$, i.e., $w \in S$ flooded $1$ in \step{b}
    of this phase.\\
    Let $P_{wv}$ be the $wv$-path identified by $v$ in \step{c}.
    Note that $P_{wv}$ is contained entirely in $G - F^*$
    so that
    $P_{wv}$ does not have any faulty nodes.
    It follows that, in \step{b},
    $w$ flooded the value $1$.
    So $v$ received value $1$ along $P_{wv}$.
    Therefore,
    in \step{c}
    $v$ puts $w$ in the set $N_v$.
  \end{enumerate}
  So we have $Z \subseteq Z_v$ and $N \subseteq N_v$, as required.
  A symmetric argument gives $Z \subseteq Z_u$ and $N \subseteq N_u$.
  As argued before, this implies $Z = Z_u = Z_v$ and $N = N_u = N_v$.
\end{proof}

We are now ready to prove \lemmaRef{lemma hypergraphs agreement}.

\begin{proof}[Proof of \lemmaRef{lemma hypergraphs agreement}]
  Consider the phase where $F = F^*$
  and $S \subseteq V - F^*$ is the unique source component
  in the directed decomposition of $G - F^*$.
  Suppose $u, v \in V - F^*$ are any two non-faulty nodes.
  Then,
  \begin{enumerate}[label=\arabic*)]
    \item
    by \lemmaRef{lemma hypergraphs alg correct partition},
    $Z = Z_u = Z_v$ and $N = N_u = N_v$,
    where $Z$ and $N$ are as in the statement of
    \lemmaRef{lemma hypergraphs alg correct partition}, and
    \item 
    by \lemmaRef{lemma algorithm correct split graph multiple},
    $G'_u = G'_v$.
  \end{enumerate}
  Let ${G'} = G'_u = G'_v$.
  We use $F'$ to denote the set of nodes in $G'$ corresponding
  to nodes in $F^*$ in $G$.

  We first show that all non-faulty nodes in $S$
  have identical state at the end of this phase.
  Observe that non-faulty nodes in $S$
  update their states exclusively in \step{d}.
  So consider \step{d} of this phase.
At the start of \step{d},
  by construction of $Z$ and $N$,
  all non-faulty nodes in $Z$ have identical state of $0$,
  while all non-faulty nodes in $N$ have identical state of $1$.
  We show that, in \step{d},
  either all non-faulty nodes in $Z$ update their state to $1$,
  or all non-faulty nodes in $N$ update their state to $0$.
  Note that ${G'} \in \graphSplitFSet_{F^*}(G)$
  and by \conditionSC{},
  either $Z \propagate{G' - N \cap F'} N - F'$
  or $N \propagate{G' - Z \cap F'} Z - F'$.
  We consider each case as follows.
  
  \begin{enumerate}
    [label={Case \arabic*:},wide,itemindent=0pt,leftmargin=2.75mm]
    \item
    $Z \propagate{G' - (N \cap F')} N - F'$.

    \noindent
    In this case,
    we show that all non-faulty nodes in $S$
    have state $0$ at the end of \step{d}.
    There are a further two cases to consider.
    \begin{enumerate}[label={Case (\roman*):},leftmargin=*]
      \item $N - F'$ is empty.\\
      Then all non-faulty nodes in $S = (Z \cup N) - F' = Z - F'$
      have state $0$ at the start of the phase.
      Each node $v \in S$
      sets $B_v = N$ in \step{d}.
      We have $B_v - F' = N - F' = \emptyset$.
      So $v$ does not update its state in \step{d}.
      Therefore, all non-faulty nodes in $S$
      have identical state $0$ at the end of \step{d}.

      \item $N - F'$ is non-empty.\\
      Consider an arbitrary node $v \in S = (Z \cup N) - F'$.
      In \step{d},
      $v$ sets $A_v = Z$ and $B_v = N$.
      If $v \in A_v - F' = Z - F'$,
      then $v$ has state $0$ at the start of this phase
      and does not update it in \step{d}.
      So suppose that $v \in B_v - F' = N - F'$.
      Now,
      if in \step{b} $v$ received the value $0$ identically
      along some $f+1$ node-disjoint $Zv$-paths in $G' - (N \cap F')$,
      then $v$ sets $\gamma_v = 0$.
      We show that such $f+1$ node-disjoint $Zv$-paths do indeed exist.
      Since $Z \propagate{G' - (N \cap F')} N - F'$,
      there exist $f+1$ node-disjoint $Zv$-paths in $G' - (N \cap F')$.
      Without loss of generality only the source nodes on these paths
      are from $Z$.
      For each such path,
      observe that only the source node,
      say $z \in Z$,
      can be faulty.
      If the source node $z$ is faulty,
      then by \lemmaRef{lemma algorithm correct split graph multiple},
      and construction of $G'$ and $Z$,
      $z$ sent the value $0$ on the first edge on this path in \step{b}.
      If $z$ is non-faulty, then by construction of $Z$,
      $z$ flooded
      value $0$ in \step{b}.
      Now all other nodes on the path are non-faulty,
      so $v$ received value $0$ along this path in \step{b}.
      Therefore,
      $v$ received value $0$ identically
      along the $f+1$ node-disjoint $Zv$-paths in \step{b},
      as required.
\end{enumerate}

    \item
    $Z \notpropagate{G' - (N \cap F')} N - F'$ so that
    $N \propagate{G' - (Z \cap F')} Z - F'$
    by \conditionSC{}.
    
    \noindent
    In this case,
    we show that all non-faulty nodes in $S$
    have state $1$ at the end of \step{d}.
    There are a further two cases to consider.
    \begin{enumerate}[label={Case (\roman*):},leftmargin=*]
      \item $Z - F'$ is empty.\\
      Then, similar to Case 1(i),
      all non-faulty nodes in $S$
      have state $1$ at the start of the phase
      and they do not update their state in \step{d}.
      So all non-faulty nodes in $S$ have state identical state $1$
      at the end of \step{d}.

      \item $Z - F'$ is non-empty.\\
      Consider an arbitrary node $v \in S = (Z \cup N) - F'$.
      In \step{d},
      $v$ sets $A_v = N$ and $B_v = Z$.
      If $v \in A_v - F' = N - F'$,
      then $v$ has state $1$ at the start of this phase
      and does not update it in \step{d}.
      So suppose that $v \in B_v - F' = Z - F'$.
      As in Case 1(ii),
      since $N \propagate{G' - (Z \cap F')} Z - F'$,
      there exist $f+1$ node-disjoint $Nv$-paths in $G' - (Z \cap F')$
      such that $v$ received the value $1$ identically along these paths
      in \step{b}.
      Therefore,
      $v$ sets $\gamma_v = 1$, as required.
    \end{enumerate}
  \end{enumerate}
  In both cases,
  all non-faulty nodes in $S$ have identical state,
  say $\tau$,
  at the end of \step{d}.
  Since nodes in $S$ do not update their state after this step,
  nodes in $S$ have state $\tau$ at the end of this phase.

  We now consider \step{f} and an arbitrary non-faulty node $v \in V - S - F^*$.
  All nodes in $S$ are non-faulty,
  so each of them floods the value $\tau$ in \step{e}.
  By 
  \lemmaRef{lemma hypergraphs source propagates},
  $S \propagate{G - F^*} V - S - F^*$
  and so there exist $f+1$ node-disjoint $Sv$-paths in
  $G - F^*$.
  All the source nodes on these paths are non-faulty nodes in $S$.
  All the internal nodes on these paths are non-faulty as well.
  So $v$ receives the value $\tau$ identically along these $f+1$
  node-disjoint paths in \step{e}.
  It follows that $v$ updates $\gamma_v$ to the value $\tau$
  in \step{f}.
  Therefore, all nodes in $V - S - F^*$
  have state $\tau$ at the end of this phase, as required.
\end{proof}

Using Lemmas
\ref{lemma hypergraphs validity}
and
\ref{lemma hypergraphs agreement},
we can now prove the sufficiency of \conditionSC{}.
Recall that
\conditionSC{} is equivalent to \conditionNC{}
by \theoremRef{theorem hypergraphs menger}.
Thus,
this shows the reverse direction of \theoremRef{theorem hypergraphs main}.

\begin{proof}[Proof of \theoremRef{theorem hypergraphs main} ($\Leftarrow$ direction)]
  \algoRef{algorithm hypergraphs} satisfies the \emph{termination}
  condition because it terminates in finite time.

  In one of the iterations of the main \texttt{for} loop,
  we have $F = F^*$,
  i.e., $F$ is the actual set of faulty nodes.
  By \lemmaRef{lemma hypergraphs agreement},
  all non-faulty nodes have the same state at the end of this phase.
  By \lemmaRef{lemma hypergraphs validity},
  these states remain unchanged in any subsequent phases.
  Therefore, all nodes output an identical state.
  So the algorithm satisfies the \emph{agreement} condition.

  At the start of phase 1,
  the state of each non-faulty node equals its own input.
  By inductively applying \lemmaRef{lemma hypergraphs validity},
  we have that the state of a non-faulty node always equals the \emph{input}
  of some non-faulty node, including in the last phase of the algorithm.
  So the output of each non-faulty node is an input of some non-faulty node,
  satisfying the \emph{validity} condition.
\end{proof}

\section{On Lemma 3 of \cite{Ravikant10.1007/978-3-540-30186-8_32}}
\label{appendix hypergraphs misc}

The bug in proof of Lemma 3 in \cite{Ravikant10.1007/978-3-540-30186-8_32}
is on the first line of page 457:
sets $C_1, C_2, C_3$ may have negative size.
Here,
we present a counter example to the claim in Lemma 3 of
\cite{Ravikant10.1007/978-3-540-30186-8_32}.
We first need the following definition of hypergraph connectivity.

\begin{definition}[Definition 4 in
  \cite{Ravikant10.1007/978-3-540-30186-8_32}]
  \label{def hyper conn}
  For $\ell, t, k > 0$,
  an undirected hypergraph $G$ is \emph{$(\ell, t)$-hyper-$k$-connected},
  if,
  for any set $C$ of \emph{exactly} $k - 1$ nodes
  and any partition of $V - C$ into $\ell$ non-empty sets,
  each of size at most $t$,
  there exists an undirected hyperedge in $G$ that has a non-empty intersection
  with every set of the partition.
\end{definition}

Recall that an undirected hyperedge $e \in E$
is a subset of nodes $e \subseteq V$ and is called an $\abs{e}$-hyperedge.
Recall also that an undirected hypergraph $G = (V, E)$ is a $(2,3)$-hypergraph
if each hyperedge is either a $2$-hyperedge or a $3$-hyperedge.
The claim in Lemma 3 of \cite{Ravikant10.1007/978-3-540-30186-8_32} is
as follows.

\begin{claim}[Lemma 3 in \cite{Ravikant10.1007/978-3-540-30186-8_32}]
  \label{claim lemma 3 Ravikant}
  An undirected $(2, 3)$-hypergraph $G = (V, E)$ with $2f < n \le 3f$
  is $(3, f)$-hyper-$(3f - n +1)$-connected
  if and only if,
  for every $V_1, V_2, V_3 \subseteq V$ such that
  $V_1 \cup V_2 \cup V_3 = V$ and $\abs{V_1} = \abs{V_2} = \abs{V_3} = f$,
  there exist three nodes
  \begin{enumerate}[label=(\roman*)]
    \item $u \in V_1 - (V_2 \cup V_3)$,
    \item $v \in V_2 - (V_1 \cup V_3)$, and
    \item $w \in V_3 - (V_1 \cup V_2)$,
  \end{enumerate}
  such that ${u, v, w} \in E$.
\end{claim}

\begin{proof}[Counter example]
  We show a counter example to the reverse direction.
  That is,
  we create an undirected $(2, 3)$-hypergraph $G = (V, E)$
  with $2f < n \le 3f$
  that satisfies both of the following:
  
  \begin{enumerate}[label=\arabic*)]
    \item 
    for every $V_1, V_2, V_3 \subseteq V$ such that
    $V_1 \cup V_2 \cup V_3 = V$ and $\abs{V_1} = \abs{V_2} = \abs{V_3} = f$,
    there exist three nodes
    \begin{enumerate}[label=(\roman*)]
      \item $u \in V_1 - (V_2 \cup V_3)$,
      \item $v \in V_2 - (V_1 \cup V_3)$, and
      \item $w \in V_3 - (V_1 \cup V_2)$,
    \end{enumerate}
    such that ${u, v, w} \in E$,

    \item
    $G$ is not $(3, f)$-hyper-$(3f - n +1)$-connected.
  \end{enumerate}

  Pick $f > 2$ and $n = 3f - 1 > 2f$.
  $G$ has all $2$-hyperedges and has two parts $X$ and $Y$.
  $X$ consists of $2f+1$ nodes and $Y$ consists of $f - 2$ nodes.
  Every 3 nodes in $X$ form a $3$-hyperedge but
  no node in $Y$ is part of any $3$-hyperedge.
  We show that $G$ satisfies each of the two condition above, as follows.

  \begin{enumerate}[label=\arabic*)]
    \item
    Consider any three sets $V_1, V_2, V_3 \subseteq V$ such that
    $V_1 \cup V_2 \cup V_3 = V$ and $\abs{V_1} = \abs{V_2} = \abs{V_3} = f$.
    By choice of $n$ and $f$ ($2f < n \le 3f$),
    such sets do exist.
    Now,
    \begin{align*}
      V_1 \cap X - (V_2 \cup V_3)
        &= X - (V_2 \cup V_3)
        &\text{since $V_1 \cup V_2 \cup V_3 = V \supseteq X$}
        \\
        &\ne \emptyset
        &\text{since $\abs{X} = 2f + 1 > 2f = \abs{V_2} + \abs{V_3} \ge \abs{V_2 \cup V_3}$.}
    \end{align*}
    Similarly,
    $V_2 \cap X - (V_1 \cup V_3) \ne \emptyset$ and
    $V_3 \cap X - (V_1 \cup V_2) \ne \emptyset$.
    It follows that
    there exist three nodes
    \begin{enumerate}[label=(\roman*)]
      \item
      $u \in V_1 \cap X - (V_2 \cup V_3) \subseteq V_1 - (V_2 \cup V_3)$,
      
      \item
      $v \in V_2 \cap X - (V_1 \cup V_3) \subseteq V_2 - (V_1 \cup V_3)$,
      and
      
      \item
      $w \in V_3 \cap X - (V_1 \cup V_2) \subseteq V_3 - (V_1 \cup V_2)$.
    \end{enumerate}
    By construction of $G$ and $X$,
    since $u, v, w \in X$,
    so $\set{u, v, w} \in E$,
    as required.

    \item
    Pick any node $x \in X$.
    Let $C := \set{x}$.
    We create a partition $(S_1, S_2, S_3)$ of $V - C$ as follows.
    $S_1$ contains exactly $f$ nodes from $X - C$.
    $S_2$ contains the remaining $f$ nodes in $X - C - S_1$.
    $S_3 := Y$.
    Then,
    \begin{align*}
      0 < \abs{S_1} &= f  \\
      0 < \abs{S_2} &= f  \\
      0 < \abs{S_3} &= f - 2 \le f.
    \end{align*}
    Since no $3$-hyperedge crosses $Y = S_3$ in $G$,
    so there is no undirected hyperedge in $G$
    that has a non-empty intersection with each of $S_1$, $S_2$, and $S_3$.
    By \defRef{def hyper conn},    
    $G$ is not $(3, f)$-hyper-$(3f - n +1)$-connected.
  \end{enumerate}

  This completes the counter example to \claimRef{claim lemma 3 Ravikant}.
\end{proof}

\end{document}